\documentclass{svjour3}
\usepackage{listings}
\usepackage{mathtools,amsmath,amsfonts}
\usepackage{amssymb,yfonts,upgreek}
\usepackage{lscape} 

\newcommand*{\id}{{\mathrm{id}}}

\newcommand*{\la}{{\langle}}                                                  
\newcommand*{\ra}{{\rangle}}                                                  
\newcommand*{\Rb}{{\mathbb R}}                                                 
\newcommand*{\Cb}{{\mathcal{C}}}                                               
                                                
\newcommand*{\Hb}{{\mathbb H}}                                                
\newcommand*{\Ub}{{\mathbb U}}                                                
\newcommand*{\Vb}{{\mathbb V}}
\newcommand*{\Wb}{{\mathbb W}}                                                
\newcommand*{\Sb}{{\mathbb S}}                                                
\newcommand*{\Nb}{{\mathbb N}}

\newcommand*{\Pb}{{\mathbb P}}

\newcommand*{\rd}{{\mathrm d}}                                                
                                               
\newcommand*{\Gr}{{\mathrm{Gr}}} 
\newcommand{\pa}{\partial}
\newcommand{\und}{\boldsymbol}
\newcommand{\underl}{}
\newcommand{\smb}{{\otimes}}
\newcommand{\mfq}{\mathfrak{q}}
\newcommand{\md}{\mathfrak{d}}
\newcommand{\tsd}{\textswab{d}}
\newcommand{\cb}{C}
\newcommand{\bb}{B}
\newcommand{\pp}{\pi}
\newcommand{\tp}{\tilde}
\newcommand{\mm}{m}
\newcommand{\wdh}{}

\smartqed
\journalname{}
\allowdisplaybreaks

\begin{document}

\title{The non-commutative Korteweg--de Vries hierarchy and combinatorial P\"oppe algebra}
\titlerunning{KdV hierarchy and P\"oppe algebra}
\author{Simon J. A. Malham$^\ast$}
\authorrunning{Simon J. A. Malham}
\institute{Maxwell Institute for Mathematical Sciences,        
and School of Mathematical and Computer Sciences,   
Heriot-Watt University, Edinburgh EH14 4AS\\ \email{S.J.A.Malham@hw.ac.uk} \\
$\ast$ The author dedicates this paper to Claudia Wulff,
a wonderful mathematician, person and brave friend, who
passed away on the 12th of June 2021.}
\date{9th August 2021}
\maketitle

\begin{abstract}
We give a constructive proof, to all orders, that each member
of the non-commutative potential Korteweg--de Vries hierarchy is
a Fredholm Grassmannian flow and is therefore linearisable.
Indeed we prove this for any linear combination of fields from this hierarchy.
That each member of the hierarchy is linearisable,
and integrable in this sense, means that the time evolving solution 
can be generated from the solution to the corresponding linear dispersion equation in the hierarchy,
combined with solving an associated linear Fredholm equation representing the Marchenko equation.
Further, we show that within the class of polynomial partial differential fields, at every order,
each member of the non-commutative potential Korteweg--de Vries hierarchy is unique.
Indeed, we prove to all orders, that each such member matches the non-commutative Lax hierarchy field,
which is therefore a polynomial partial differential field.
We achieve this by constructing the abstract combinatorial algebra that underlies
the non-commutative potential Korteweg--de Vries hierarchy.
This algebra is the non-commutative polynomial algebra over the real line generated
by the set of all compositions endowed with the P\"oppe product.
This product is the abstract representation of the product rule for Hankel operators pioneered
by Ch.\/ P\"oppe for integrable equations such as the Sine-Gordon and Korteweg--de Vries equations.  
Integrability of the hierarchy members translates, in the combinatorial algebra, to proving
the existence of a `P\"oppe polynomial' expansion for basic compositions in terms of
`linear signature expansions'.
Proving the existence of such P\"oppe polynomial expansions boils down to solving a linear algebraic
problem for the expansion coefficients, which we solve constructively to all orders.
\end{abstract}

\section{Introduction}\label{sec:intro}
We prove, to all orders, that each member of
the non-commutative potential Korteweg--de Vries hierarchy is a Fredholm Grassmannian flow.
As such, each member of the hierarchy is linearisable, and therefore integrable in this sense.
The proof straightfowardly extends to any linear combination of fields from this hierarchy.
In addition we show that within the class of non-commutative polynomial partial differential fields,
each member of the hierarchy is unique.
Indeed, we prove that each such member matches the non-commutative Lax hierarchy field,
which at every order is therefore a non-commutative polynomial partial differential field.
By `field' here we mean the right-hand side in the partial differential equation when
we express it in the form of the derivative in time of the solution equals the linear
and nonlinear terms on the right. By a `non-commutative polynomial partial differential field'
we mean that the field can expressed as a non-commutative multivariate polynomial whose 
arguments are the solution and spatial partial derivatives of the solution only. 
When we say each  member of the hierarchy is linearisable, we mean that the time-evolving solution 
can be generated by solving the corresponding linear dispersion equation in the hierarchy for
the `scattering data', and the solution to the hierarchy is then generated by solving 
an associated linear Fredholm equation involving the scattering data that represents a Marchenko equation.
Ch.\/ P\"oppe pioneered this approach to finding solutions to classical scalar integrable systems such as the 
Sine-Gordon equation, Korteweg--de Vries equation and its scalar hierarchy,
as well as the nonlinear Schr\"odinger equation in a sequence of papers.
See P\"oppe~\cite{P83,P84,P-KP}, P\"oppe and Sattinger~\cite{PS88} and Bauhardt and P\"oppe~\cite{BP-ZS}.
Recently Doikou \textit{et al.\/} \cite{DMSW20} streamlined P\"oppe's approach, demonstrating that 
only P\"oppe's celebrated kernel product rule is required to establish the linearisation
of the classical Korteweg--de Vries and nonlinear Schr\"odinger equations.
Subsequently Doikou \textit{et al.\/} \cite{DMS20} demonstrated the approach,
as considered by Bauhardt and P\"oppe~\cite{BP-ZS}, naturally extends to
the non-commutative nonlinear Schr\"odinger and Korteweg--de Vries equations.
Malham~\cite{M-quinticNLS} then used this approach to establish the linearisation
of the non-commutative fourth order quintic nonlinear Schr\"odinger equation.
The results herein extend P\"oppe's approach to the whole non-commutative Korteweg--de Vries hierarchy.
However the algebraic structures we develop to establish the linearisation of this hierarchy,
provide a deeper insight into the `integrability' properites of the hierarchy, not
only establishing uniqueness, but also transporting the question of integrability into
the existence of polynomial expansions in a given combinatorial algebra equipped with
an abstract version of the P\"oppe product. In honour of Ch.\/ P\"oppe, we
call this combinatorial algebra the \emph{P\"oppe algebra} and the polynomial expansions,
\emph{P\"oppe polynomials}.

Let us now explain the approach in some more detail. Consider a non-commutative nonlinear 
partial differential equation for $g=g(x,t)$ of the form:
\begin{equation*}  
\pa_tg=d(\pa)g+\pi(g,\pa g,\pa^2 g,\ldots),
\end{equation*}
where $\pa=\pa_x$. Here the `field' which is the right-hand side consists
of the linear $d(\pa)g$ terms,
where $d(\pa)$ is a constant coefficient odd homogeneous polynomial 
with the lowest degree term of degree $3$, and the nonlinear terms $\pi$.
We suppose $\pi$ is a precise homogeneous non-commutative polynomial
function of the arguments indicated up to including derivatives of $g$ of order $\text{deg}(d)-2$.
We assume we have absorbed the linear terms in $\pi$
into $d(\pa)g$, so the lowest degree term in $\pi$ is $2$.
Now suppose the matrix valued function $p=p(x,t)$ satisfies the corresponding
linear partial differential equation:
\begin{equation*}  
\pa_tp=d(\pa)p. 
\end{equation*}
The quantity $p=p(x,t)$ represents the so-called `scattering data'.
P\"oppe elevated the Marchenko equation to the operator level,
associating a Hilbert--Schmidt Hankel operator $P=P(x,t)$ with
$p$ as follows. For any square integrable function $\phi$ set:
\begin{equation*}  
\bigl(P\phi\bigr)(y;x,t)\coloneqq\int_{-\infty}^0p(y+z+x;t)\phi(z)\,\mathrm{d}z. 
\end{equation*}
The operator $P=P(x,t)$ satisfies the linear operator differential equation,
\begin{align*}
\pa_tP&=d(\pa)P,
\intertext{which we augment with the following linear Fredholm equation for $G=G(x,t)$,}
P&=G(\id-P).
\end{align*}
This represents the Marchenko equation at the operator level.
These are the only ingredients we need for the non-commutative potential Korteweg--de Vries hierarchy.

P\"oppe's kernel product rule involving Hankel operators is then key. 
Suppose $F=F(x,t)$ and $F^\prime=F^\prime(x,t)$ are both Hilbert--Schmidt operators
that are continuously dependent on the parameters $x$ and $t$.
Further suppose that $H=H(x,t)$ and $H^\prime=H^\prime(x,t)$ are Hilbert--Schmidt \emph{Hankel} operators,
that are continuously differentiably dependent on $x$ and $t$.
Let $[F]$ denote the kernel of any Hilbert--Schmidt operator $F$.  
Then the Fundamental Theorem of Calculus implies \emph{P\"oppe's kernel product rule}:
\begin{equation*}  
[F\pa_x(HH^\prime)F^\prime](y,z;x,t)=[FH](y,0;x,t)[H^\prime F^\prime](0,z;x,t).
\end{equation*}
With this in hand, we just compute $\pa_tG-d(\pa)G$, where $G$ satisfies the linear
Fredholm equation above, and apply the kernel bracket operator $[\,\cdot\,]$.
If we write $G=PU$ where $U\coloneqq(\id-P)^{-1}$, then we see that basic
calculus properties such as $\pa U=U(\pa P)U$ generate nonlinear terms.
The \emph{goal} is then to use \emph{only} the kernel product rule to establish
`closed forms' for the nonlinear terms generated. By this we mean that the terms generated
can be expressed as a constant coefficient non-commutative poylnomial
in $\pa[G]$, $\pa^2[G]$ and so forth. Note we do not include $[G]$ itself 
in this polynomial as our aim is to derive the
non-commutative \emph{potential} Korteweg--de Vries hierarchy equations. 

Fleshing this out further, since $U\coloneqq(\id-P)^{-1}$, we know $U\equiv\id+UP$.
We have seen that $\pa U=U(\pa P)U$. The \emph{general Leibniz rule} implies:
\begin{equation*}
\pa^n U=U(\pa^n P)U+n\,(\pa U)(\pa^{n-1}P)U+\cdots+n\,(\pa^{n-1}U)(\pa P)U,
\end{equation*}
where we have combined the $\pa^n U$ terms appearing in the Leibniz expansion,
and used the definition for $U$. 
Suppose with $n\geqslant3$ and odd, we take $d(\pa)=\mu_n\pa^n$ with $\mu_n$ a real constant.
Then since $G=PU$ and $U\equiv\id+UP\equiv\id+PU$, we observe
$\pa_tG=\pa_tU=U(\pa_t P)U=\mu_n\,U(\pa^n P)U$ and so,
\begin{align*}
&&\pa_tG&=\mu_n\bigl(\pa^n U-n\,(\pa U)(\pa^{n-1}P)U-\cdots-n\,(\pa^{n-1}U)(\pa P)U\bigr)\\
\Rightarrow&&\pa_t[G]&=\mu_n\pa^n[G]-\mu_n\Bigl(n\,\bigl[(\pa U)(\pa^{n-1}P)U\bigr]
+\cdots+n\,\bigl[(\pa^{n-1}U)(\pa P)U\bigr]\Bigr),
\end{align*}
where we have used that $\pa^nG=\pa^nU$ and so $\pa^n[G]=[\pa^n U]$. 
The question is, can we express all the remaining terms on the right shown in the desired `closed form'?
In other words, can we express them as a constant coefficient non-commutative poylnomial
in $\pa[G]$, $\pa^2[G]$ and so forth. Note that $\mu_n\bigl[U(\pa^n P)U\bigr]=\pa_t[G]$
and from the general Leibniz rule above, we have,
\begin{align*}
  \bigl[U(\pa^n P)U\bigr]
  =\pa^n[U]-n\,\bigl[(\pa U)(\pa^{n-1}P)U\bigr]-\cdots-n\,\bigl[(\pa^{n-1}U)(\pa P)U\bigr].
\end{align*}
Thus an \emph{equivalent question} is, can we express all the terms on the right in terms of
a constant coefficient non-commutative poylnomial in $\pa[U]$, $\pa^2[U]$ and so forth?
This is obviously trivial in the case of the first term on the right.

Let us consider a quick example, naturally for the case $n=3$. It is possible to
show in a few lines, see Example~\ref{ex:kdV} in Section~\ref{sec:existenceanduniqueness}, that 
\begin{equation*}
\bigl[U(\pa^3P)U\bigr]=\pa^3[U]-3\,\bigl(\pa[U]\bigr)\bigl(\pa[U]\bigr).
\end{equation*}
Note, from the P\"oppe kernel product rule above,
when we write $\bigl(\pa[U]\bigr)\bigl(\pa[U]\bigr)$ we mean 
$\bigl(\pa[U]\bigr)(y,0;x,t)\,\bigl(\pa[U]\bigr)(0,z;x,t)$.
This thus establishes that $[G]=[G](y,z;x,t)$ satisfies the 
partial differential equation,
\begin{equation*}
\pa_t[G]=\mu_3\pa^3[G]-3\mu_3\,\bigl(\pa[G]\bigr)\bigl(\pa[G]\bigr).
\end{equation*}
If we set $y=z=0$, then as a special case, $g(x,t)\coloneqq[G](0,0;x,t)$ satisfies
\begin{equation*}
\pa_tg=\mu_3\pa^3g-3\mu_3\,(\pa g)(\pa g),
\end{equation*}
which is the non-commutative potential Korteweg--de Vries equation. 

It is the \emph{equivalent question} above that is crucial to tackling
the higher order cases systematically. Returning to the original
\emph{general Leibniz rule} for operators above, we observe 
the expression for $\pa^n U$ contains all the lower order derivatives
$\pa^{n-1}U$,\ldots,$\pa U$ as factors in the terms on the right.
We can insert the corresponding Leibniz rules for the lower order
terms. This generates the expansion,
see Theorems~\ref{thm:invOpII} and \ref{thm:underlinepolyexpansion}:
\begin{equation*}
\pa^n U=\sum\chi\bigl(a_1a_2\cdots a_k\bigr)
UP_{a_1}UP_{a_2}U\cdots UP_{a_k}U,
\end{equation*}
where the sum is over all compositions $a_1a_2\cdots a_k\in\Cb(n)$ of $n$, and $P_k\coloneqq\pa^kP$.
We call the real coefficient $\chi\bigl(a_1a_2\cdots a_k\bigr)$ 
the \emph{signature character}, it has an explicit form as a product of Leibniz coefficients;
see Definition~\ref{def:signaturecharacter}.

Applying the kernel bracket to this expression we naturally get,
\begin{equation*}
[\pa^n U]=[UP_nU]+\sum\chi\bigl(a_1a_2\cdots a_k\bigr)\,
[UP_{a_1}UP_{a_2}U\cdots UP_{a_k}U],
\end{equation*}
where the sum is over all compositions $a_1a_2\cdots a_k\in\Cb(n)\backslash\{n\}$; we exclude $n$ itself
which we have distinguished in the first term and used that $\chi(n)=1$. 
Further, in this context, P\"oppe's kernel product rule is equivalent to
the following statement. If $u=a_1\cdots a_k$ and $v=b_1\cdots b_\ell$ are compositions 
and $a$ and $b$ are natural numbers, then denoting 
$\mm_u\coloneqq UP_{a_1}UP_{a_2}U\cdots UP_{a_k}U$ and
$\mm_v\coloneqq UP_{b_1}UP_{b_2}U\cdots UP_{b_\ell}U$, we have 
\begin{equation*}
[\mm_uP_aU][UP_b\mm_v]
=[\mm_uP_{a+1}UP_b\mm_v]+[\mm_uP_aUP_{b+1}\mm_v]+2\,[\mm_uP_aUP_1UP_b\mm_v].
\end{equation*}
Now let us take stock. We see that we can always express the terms $[\pa^n U]$
for any $n\in\Nb$ as a linear combination of monomials of the form $[UP_{a_1}UP_{a_2}U\cdots UP_{a_k}U]$.
Further we have an explicit expression for the product of any such monomials.
Indeed, we observe that in principle we can focus our attention entirely on an algebra
of such monomials in which the product is the one just stated.
The terms $[\pa^n U]$ are simple, natural example linear combinations of monomials in this algebra.
We call them the \emph{linear signature expansions}, or \emph{signature expansions} for short.
Now, rearranging the expression for $[\pa^n U]$ above and noting that $\chi(n)=1$, we find, 
\begin{equation*}
[UP_nU]=[\pa^n U]-\sum\chi\bigl(a_1a_2\cdots a_k\bigr)\,
[UP_{a_1}UP_{a_2}U\cdots UP_{a_k}U],
\end{equation*}
where the sum is over all compositions $a_1a_2\cdots a_k\in\Cb(n)\backslash\{n\}$.
And now the \emph{equivalent question} we asked previously above becomes, 
can we express $[UP_nU]$, or equivalently all the terms on the right,
in terms of a constant coefficient non-commutative polynomial in
$[\pa U]$, $[\pa^2 U]$, and so forth?

We observe that the variable labels `$P$' and `$U$' are superfluous.
In our discussion above, we ultimately consider the real algebra of the monomials
$[UP_{a_1}U\cdots UP_{a_k}U]$ where the product of any two such monomials
is given by the P\"oppe kernel product. We can replace these monomials directly
by compositions $a_1a_2\cdots a_k\in\Cb$, with the corresponding P\"oppe
product `$\ast$' given as follows. For two compositions $ua$ and $bv$ in $\Cb$,
where we distinguish the last letter and first letters $a$ and $b$ as indicated,
the P\"oppe product for compositions is given by
\begin{equation*}
(ua)\ast(bv)=u(a+1)bv+ua(b+1)v+2\cdot(ua1bv).
\end{equation*}
This mirrors the P\"oppe kernel product rule above. We further observe that
we can confine ourselves to computations in the algebra $\Rb\la\Cb\ra_\ast$,
the real algebra over $\Cb$ endowed with the product `$\ast$'. 
We define the linear signature expansions in this context as the elements
$\und{n}\in\Rb\la\Cb\ra_\ast$ given for any $n\in\Nb$ by, 
\begin{equation*}
\wdh{\und{n}}=\sum_{w\in\Cb(n)}\chi(w)\cdot w.
\end{equation*}
These correspond to the terms $\pa^n[U]$.
Naturally the product of any two linear signature expansions 
$\wdh{\und{m}}$ and $\wdh{\und{n}}$ in $\Rb\la\Cb\ra_\ast$
is given by, $\wdh{\und{m}}\ast\wdh{\und{n}}=\sum\chi(u)\chi(v)\cdot u\ast v$,
where the sum is over all $u\in\Cb(m)$ and $v\in\Cb(n)$.
Further from our coding just above, the single letter composition $n\in\Rb\la\Cb\ra_\ast$
corresponds to the term $[UP_nU]$.
Now the \emph{equivalent question} above becomes: for every odd $n\in\Nb$, can we find
a polynomial $\pi_n=\pi_n\bigl(\und{1},\und{2},\ldots,\und{(n-2)},\und{n}\bigr)$ such
that $\pi_n\bigl(\und{1},\und{2},\ldots,\und{(n-2)},\und{n}\bigr)=n$ where the
polynomial $\pi_n$ has the form:
\begin{equation*}
\pi_n\coloneqq\sum_{k=1}^{\frac12(n+1)}\sum_{a_1a_2\cdots a_k\in\Cb^\ast(n)}
c_{a_1a_2\cdots a_k}\cdot\wdh{\und{a_1}}\ast\wdh{\und{a_2}}\ast\cdots\ast\wdh{\und{a_k}},
\end{equation*}
where $\Cb^\ast(n)\subset\Cb(n)$ represents the subset
of compositions $w=a_1a_2\cdots a_k$ of $n$ such that
$a_1+a_2+\cdots+a_k=n-k+1$. This is the appropriate subset
of $k$-factor monomials of the form
$\wdh{\und{a_1}}\ast\wdh{\und{a_2}}\ast\cdots\ast\wdh{\und{a_k}}$
that generate compositions in $\Cb(n)$ once their P\"oppe products are evaluated.
The coefficients $c_{a_1a_2\cdots a_k}$ are real constants.
These polynomials are the \emph{P\"oppe polynomials} mentioned earlier.
To determine the coefficients $c_{a_1a_2\cdots a_k}$, we expand all the
P\"oppe products in all the monomials
$\wdh{\und{a_1}}\ast\wdh{\und{a_2}}\ast\cdots\ast\wdh{\und{a_k}}$
present in $\pi_n$. Equating the coefficients of all the compositions in $\Cb(n)$
thus generated, results in an overdetermined system of linear equations for
the coefficients $c_{a_1a_2\cdots a_k}$ for all $a_1a_2\cdots a_k\in\Cb^\ast(n)$.
The right-hand side of the equation $\pi_n=n$ is simply $n\in\Rb\la\Cb\ra_\ast$
and thus all the linear equations are homogeneous apart from that corresponding
to the single letter composition $n\in\Rb\la\Cb\ra_\ast$,
i.e.\/ corresponding to the coefficient $c_n$.
We show there is a unique solution to this linear system and thus 
the corresponding nonlinear partial differential equation in the hierarchy
is integrable. In Section~\ref{sec:Laxhierarchy}, at each odd order $n\in\Nb$,
we derive the non-commutative Lax hierarchy from scratch in the algebra $\Rb\la\Cb\ra_\ast$.
Using a technique analogous to that just outlined, we show that for every odd $n\in\Nb$,
the polynomial expansions for the single letter compositions $n\in\Rb\la\Cb\ra_\ast$
generated by the non-commutative iterative Lax procedure can also be expressed 
as a polynomial expansion in terms of $\und{1}$, $\und{2}$, \ldots, $\und{(n-2)}$ and $\und{n}$.
Since we previously proved such an expansion for the single letter compositions $n\in\Rb\la\Cb\ra_\ast$
is unique, the non-commutative nonlinear partial differential equations generated
by the non-commutative Lax hierarchy, and those we generated using the P\"oppe polynomials
just above, are \emph{one and the same}. Further, as part of this whole procedure,
we also establish a new co-algebra, the `signature co-algebra'.

The algebraic approach to the non-commutative Korteweg--de Vries hierachy
introduced herein, necessarily arose from the desire to extend the Grassmannian flow
solution approach for classical integrable systems. This approach has its roots
in the work of Dyson~\cite{Dyson}, Ablowitz \textit{et al.\/} \cite{ARS} and Mumford~\cite{Mumford},
and in particular in a series of papers by 
P\"oppe~\cite{P83,P84,P-KP}, P\"oppe and Sattinger~\cite{PS88} and Bauhardt and P\"oppe~\cite{BP-ZS}.
Other fundamental approaches include the scheme by Zakharov and Shabat~\cite{ZS,ZS2}
as well as, for example, the unified transform method, see Fokas and Pelloni~\cite{FP}.
P\"oppe's approach was recently revisited by McKean~\cite{McKean}. 
Combining the Grassmannian flow solution approach for 
infinite dimensional Riccati partial differential systems
developed in Beck~\textit{et al.} \cite{BDMSI,BDMSII},
with the Hankel operator approach for classical integrable systems
developed by P\"oppe, Doikou \textit{et al.} \cite{DMSW20}
showed how for general initial data, time evolutionary solutions
to the Korteweg--de Vries and nonlinear Schr\"odinger equations
represent Grassmannian flows. Such flows are linearisable in the 
sense outlined above. Stylianidis~\cite{Stylianidis}
and Doikou \textit{et al.\/} \cite{DMS20} extended this approach
to the non-commutative versions of these classical integrable systems.
Malham~\cite{M-quinticNLS} extended the approach to the non-commutative
fourth order quintic nonlinear Schr\"odinger equation.
Early work on non-commutative integrable systems includes
that by Fordy and Kulisch~\cite{FK}, Ablowitz \textit{et al.} \cite{APT}
and Ercolani and McKean~\cite{EM}.
Work by Nijhoff \textit{et al.} \cite{NQLC}, Fu~\cite{F} and Fu and Nijhoff~\cite{FNI}
and the direct linearisation technique they develop is extremely close in
spirit to the approach we develop herein. Even closer, in the sense that
they elevate the linearisation solution technique to the operator
level and seek solutions in the form of inverse Fredholm operators
associated with scattering data exactly as we do, is the recent work
by Carillo and Schoenlieb~\cite{CSI,CSIa,CSII}. Also see Aden and Carl~\cite{AdenCarl},
Hamanaka and Toda~\cite{HamanakaToda}, as well as Sooman~\cite{Sooman}.
The work by Treves~\cite{TI,TII} is also close in the sense of
seeking to represent the nonlinear fields in the Korteweg--de Vries
hierarchy as polynomials in a given algebra of formal pseudo-differential symbols.
Also see Buryak and Rossi~\cite{BuryakRossi}. 
General algebraic solution approaches to the Kadomtsev--Petviashvili hierarchies
in the context of weakly nonassociative algebras, which are close to
the approach we adopt herein, can be found in
Dimakis and M\"uller--Hoissen~\cite{DM-H2008}, whereas
Dimakis and M\"uller--Hoissen~\cite{DM-H2005} consider
closer connections to shuffle and Rota--Baxter algebras.
See Reutenauer~\cite{Reu}, Malham and Wiese~\cite{MW},
Ebrahimi--Fard \textit{et al.\/} \cite{E-FLMM-KW}
and also Ebrahimi--Fard \textit{et al.\/} \cite{E-FMPW} for more
on shuffle algebras and references therein for Rota-Baxter algebras.
For very recent work on the non-commutative Korteweg--de Vries equation,
see Degasperis and Lombardo~\cite{DL2}, Pelinovksy and Stepanyants~\cite{Pelinovsky}
and Adamopoulou and Papamikos~\cite{AP}.
Hankel operators have also received a lot of recent attention,
see for example Grudsky and Rybkin~\cite{GR1,GR2}, Grellier and Gerard~\cite{Gerard},
and Blower and Newsham~\cite{BN}.
The connection between Grassmannians and classical integrable systems
was first explored by Sato~\cite{SatoI,SatoII} and developed further by Segal and Wilson~\cite{SW}.
Also see Kasman~\cite{Kasman1995,Kasman1998} and Hamanaka and Toda~\cite{HamanakaToda}.

To summarise, what we achieve herein is that,
for the non-commutative potential Korteweg--de Vries hierarchy:

\begin{enumerate}
\item[(i)] We show that establishing integrability at any odd order $n$ is equivalent to determining the
existence of certain polynomial expansions for the single letter composition $n$,
in the real algebra of compositions equipped with the P\"oppe product,
in terms of natural so-called (linear) signature expansions of compositions. 
\end{enumerate}

\noindent Using this algebraic structure, we prove to all orders that each member of the hierarchy:

\begin{enumerate}
\item[(ii)] Is a Fredholm Grassmannian flow and thus linearisable, and so integrable in this sense.
We establish this for any linear combination of fields from the hierarchy;

\item[(iii)] Is unique within the class of polynomial partial differential fields;

\item[(iv)] Matches the non-commutative Lax hierarchy field, which is therefore
a polynomial partial differential field.
\end{enumerate}

\noindent We also:

\begin{enumerate}
\item[(v)] Establish a new co-algebra we call the signature co-algebra based on the de-P\"oppe
co-product, which tensorially decomposes any composition into the sum of all possible composition
pairs that produced that composition via the P\"oppe product.
\end{enumerate}

Our paper is organised as follows.
In Section~\ref{sec:Grassmannianflows} we introduce P\"oppe's formulation for
integrable system via Hankel operators. We discuss why solutions generated
using this approach are example Fredholm Grassmannian flows.
We begin our process of abstraction in Section~\ref{sec:kernelalgebras},
showing how all the inherent operators in P\"oppe's formulation, including
the key operator corresponding to the inverse Fredholm operator associated with the scattering data,
can be expressed as an expansion in the concatenation algebra of compositions over the real
field. After applying the kernel bracket operator to this concatenation
algebra, we show how, necessariy, we need to consider the real algebra of compositions
equipped with the abstract form of the P\"oppe product.  
A short Section~\ref{sec:signaturecoalgebra} follows, in which we 
introduce the signature co-algebra. Though strictly not required
for the subsequent analysis and proofs in the sections which follow,
our motivation and the results therein help elucidate some of the constructs
we use and rely on in them.
Our main result is stated and proved in Section~\ref{sec:existenceanduniqueness}.
Indeed, at each odd order we prove the existence of a P\"oppe polynomial
expansion and thus the integrability of the corresponding non-commutative
nonlinear partial differential equation.
We generate the non-commutative Lax hierarchy in Section~\ref{sec:Laxhierarchy}.
We prove that the non-commutative hierarchy of equations
we established in the preceding section,
correspond precisely to those in the non-commutative Lax hierarchy.
Lastly in Section~\ref{sec:conclusion}, we explore extensions
of the results we establish herein, and in particular, briefly,
the application of this approach to the non-commutative nonlinear Schr\"odinger hierarchy.

\section{P\"oppe formulation and Grassmannian flows}\label{sec:Grassmannianflows}
Let us begin by defining and deriving the main analytical ingredients we need.
We consider Hilbert--Schmidt integral operators which depend on both a 
spatial parameter $x\in\mathbb{R}$ and a time parameter $t\in[0,\infty)$.
In this section $\pa_t$ represents the partial derivative with respect to
the time parameter $t$ while $\pa=\pa_x$ represents the partial derivative
with respect to the spatial parameter $x$. 
For any Hilbert--Schmidt operator $F=F(x,t)$ there exists a square-integrable
kernel $f=f(y,z;x,t)$ such that for any square-integrable function $\phi$,
\begin{equation*}
  (F\phi)(y;x,t) = \int_{-\infty}^0 f(y,z;x,t)\phi(z)\,\mathrm{d}z.
\end{equation*}
\begin{definition}[Kernel bracket]
With reference to the operator $F$ just above, we use the \emph{kernel bracket} notation
$[F]$ to denote the kernel of $F$:
\begin{equation*}
  [F](y,z;x,t) \coloneqq f(y,z;x,t).
\end{equation*}
\end{definition}
Hilbert-Schmidt Hankel operators play a critical role herein. 
We consider Hankel operators which depend on a parameter $x$ as follows. 
\begin{definition}[Hankel operator with parameter]\label{def:Hankel}
We say a given time-dependent Hilbert--Schmidt operator $H$
with corresponding square-integrable kernel $h$ is \emph{Hankel} or \emph{additive}
with parameter $x\in\Rb$ if its action, for any square-integrable function $\phi$, is 
\begin{equation*}
  (H\phi)(y;x,t) \coloneqq \int_{-\infty}^0 h(y+z+x;t)\phi(z)\,\mathrm{d}z.
\end{equation*}
\end{definition}
Such Hankel operators are a key facet of P\"oppe's formulation.
This is because the kernel of the spatial derivative $\pa=\pa_x$ of the operator product
of an arbitrary pair of such Hankel operators can be split into the real matrix product of their respective
kernels in the following sense; see P\"oppe~\cite{P83,P84}. We include the proof from
Doikou \textit{et al.\/} \cite{DMSW20,DMS20} and Malham~\cite{M-quinticNLS} for completeness.
\begin{lemma}[P\"oppe product]\label{lemma:kernelproductrule}
Assume for each $n\in\mathbb N$, $H_n$ are Hankel Hilbert--Schmidt operators with parameter $x$
and $F_n$ are Hilbert--Schmidt operators.
Assume further for each $n\in\mathbb N$, that the corresponding kernels of $F_n$ are continuous and
those of $H_n$ are continuously differentiable. Then, the following \emph{P\"oppe product} rule holds,
\begin{multline*}
[F_1\pa(H_1H_2)F_2F_3\pa(H_3H_4)F_4\cdots F_{n-1}\pa(H_{n-1}H_n)F_n](y,z;x)\\
=[F_1H_1](y,0;x)[F_2H_2](0,0;x)\cdots [F_{n-1}H_{n-1}](0,0;x)[H_nF_n](0,z;x).
\end{multline*}
\end{lemma}
\begin{proof}
It is sufficient to establish the result for the case $n=2$. The general
result follows by successive iteration of this case.
We use the fundamental theorem of calculus and Hankel properties of $H_1$ and $H_2$.
Let $f_1$, $h_1$, $h_2$ and $f_2$ denote the integral kernels
of $F_1$, $H_1$, $H_2$ and $F_2$ respectively.
By direct computation $[F_1\pa_x(H_1H_2)F_2](y,z;x)$ equals
\begin{align*}
&\int_{\Rb_-^3}
f_1(y,\xi_1;x)\pa_x\bigl(h_1(\xi_1+\xi_2+x)h_2(\xi_2+\xi_3+x)\bigr)
f_2(\xi_3,z;x)\,\rd \xi_3\,\rd \xi_2\,\rd \xi_1\\
&=\int_{\Rb_-^3}
f_1(y,\xi_1;x)\pa_{\xi_2}\bigl(h_1(\xi_1+\xi_2+x)h_2(\xi_2+\xi_3+x)\bigr)
f_2(\xi_3,z;x)\,\rd \xi_3\,\rd \xi_2\,\rd \xi_1\\
&=\int_{\Rb_-^2}
f_1(y,\xi_1;x)h_1(\xi_1+x)h_2(\xi_3+x)f_2(\xi_3,z;x)\,\rd \xi_3\,\rd \xi_1\\
&=\int_{\Rb_-}f_1(y,\xi_1;x)h_1(\xi_1+x)\,
\rd \xi_1\cdot\int_{\Rb_-}h_2(\xi_3+x)f_2(\xi_3,z;x)\,\rd \xi_3\\
&=\bigl([F_1H_1](y,0;x)\bigr)\bigl([H_2F_2](0,z;x)\bigr),
\end{align*}
giving the result. \qed
\end{proof}
The following lemma and theorem, based on the Leibniz rule, are key to the abstract formalism we introduce
in Section~\ref{sec:operatoralgebras}. Both results are stated in terms of the operator $P$, though
neither result requires $P$ to be a Hankel operator.
\begin{lemma}[Inverse operator Leibniz rule]\label{lemma:invOpI} 
Suppose the operator $P$ depends on a parameter with respect
to which we wish to compute derivatives. Further suppose $U\coloneqq (\mathrm{id}-P)^{-1}$ exists.
Then we observe 
\begin{equation*}
U\equiv\id+UP\equiv\id+PU,
\end{equation*}
and also that $\pa U\equiv U(\pa P)U$, and more generally that 
\begin{equation*}
\pa^n U\equiv\sum_{k=0}^{n-1}\begin{pmatrix} n\\ k\end{pmatrix} (\pa^kU)(\pa^{n-k}P)U
\equiv\sum_{k=0}^{n-1}\begin{pmatrix} n\\ k\end{pmatrix} U(\pa^{n-k}P)(\pa^kU).
\end{equation*}
\end{lemma}
\begin{proof}
The first three identities are straightforward. The final identity follows by
applying the Leibniz rule to the first pair of identities and using $U\coloneqq(\id-P)^{-1}$.  \qed
\end{proof}
Similar identities to those in Lemma~\ref{lemma:invOpI} were derived by P\"oppe \cite{P83,P84}.
For any $n\in\mathbb N$, let $\Cb(n)$ denote the set of all compositions of $n$.
The real-valued coefficients $\chi\bigl(a_1a_2\cdots a_k\bigr)$ given in
the following theorem are essentially particular products of Leibniz coefficients, they
are given in Definition~\ref{def:signaturecharacter}.
\begin{theorem}[Signature operator expansion]\label{thm:invOpII}
For any non-negative integer $k$, set $P_k\coloneqq\pa^kP$ and $U_k\coloneqq\pa^kU$.
Then for any $n\in\mathbb N$ we have
\begin{equation*}
U_n=\sum\chi\bigl(a_1a_2\cdots a_k\bigr)
UP_{a_1}UP_{a_2}U\cdots UP_{a_k}U,
\end{equation*}
where the sum is over all compositions $a_1a_2\cdots a_k\in\Cb(n)$.
\end{theorem}
We prove this result, which is a direct consequence of Lemma~\ref{lemma:invOpI},
in a more abstract context in the next section;
see Theorem~\ref{thm:underlinepolyexpansion}.
Indeed, in Sections~\ref{sec:operatoralgebras} and \ref{sec:kernelalgebras}
we pull back the definitions and results above
to two natural abstract algebras.

Let us now outline the solution procedure based on P\"oppe's Hankel operator
approach we employ herein for the non-commutative Korteweg--de Vries hierarchy,
and its relation to Grassmannian flows.
Solutions to this linear system generate solutions to the target
nonlinear partial differential equations from the 
non-commutative potential Korteweg--de Vries hierarchy or any linear
combination of fields from the hierarchy.
\begin{definition}[Linear operator system]\label{def:linearsystem}
Suppose the linear operators $P=P(x,t)$ and $G=G(x,t)$ satisfy the pair of linear equations:
\begin{equation*}  
\pa_tP=d(\pa)P
\qquad\text{and}\qquad 
P=G(\id-P),
\end{equation*}
where $d(\pa)=\mu_3\pa^3+\mu_5\pa^5+\mu_7\pa^7+\cdots$ 
and each $\mu_n$, with $n\geqslant 3$ odd, is a constant real parameter.
\end{definition}
This system of equations for the unknown operator $G$ is linear since $P=P(x,t)$
is the solution of a linear partial differential equation and then $G=G(x,t)$ is
the solution of a linear Fredholm equation. Let us now address the existence and
uniqueness of such a solution operator in the class of Hilbert--Schmidt operators.
Suppose the square matrix-valued function $p=p(x,t)$ satisfies the linear partial 
differential equation,
\begin{equation*}  
\pa_tp=d(\pa)p,
\end{equation*}
with $p(x,0)=p_0(x)$ and $p_0$ a given square matrix-valued function. 
For $w:\Rb\to\mathbb{R}_+$, let $L^2_w$ denote the space of real, square matrix-valued functions $f$
on $\Rb$ whose $L^2$-norm weighted by $w$ is finite, i.e.\/
\begin{equation*}
  \|f\|_{L^2_w}\coloneqq\int_{\Rb} \mathrm{tr}\,\bigl(f^{\mathrm{T}}(x)f(x)\bigr)w(x)\,\mathrm{d}x<\infty,
\end{equation*} 
where $f^{\mathrm{T}}$ denotes the transpose of $f$ and `$\mathrm{tr}$' is the trace operator. 
Let $W:\mathbb{R}\to\mathbb{R}_+$ denote the function $W:x\mapsto 1+x^2$.
Further let $H$ denote the Sobolev space of real, square matrix-valued functions who themselves,
as well as derivatives $\pa$ to all orders of them, are square-integrable.
Doikou \textit{et al.\/} \cite[Lemma~3.1]{DMS20} establish if 
$p_0\in H\cap L^2_W$ then $p\in C^\infty\bigl([0,\infty);H\cap L^2_W\bigr)$ 
and the corresponding Hankel operator $P=P(t)$ whose kernel is $p$,
is Hilbert--Schmidt valued. For Hilbert--Schmidt valued operator $P$
we define the regularised Fredholm determinant to be,
\begin{equation*}
\mathrm{det}_2(\id-P)\coloneqq\exp\Biggl(-\sum_{k\geqslant 2}\frac{1}{k}\mathrm{tr}\,(P^k)\Biggr).
\end{equation*}
See Simon~\cite{Simon:Traces}. The operator $\id-P$ is invertible if and only if
$\mathrm{det}_2(\id-P)\neq0$.
Now assunme as just above, $p_0\in H\cap L^2_W$, and further that the Hilbert--Schmidt
Hankel operator $P_0$ corresponding to the kernel $p_0$ is such that $\mathrm{det}_2(\id-P_0)\neq0$.
Then from Doikou \textit{et al.\/} \cite[Lemma~3.2]{DMS20} we know, there exists
a time $T>0$, such that for all $t\in[0,T]$ and $x\in\Rb$, we know 
$\mathrm{det}_2\bigl(\id-P(x,t)\bigr)\neq0$, where $P=P(x,t)$ is the
Hilbert--Schmidt Hankel operator whose kernel is $p=p(y+x;t)$ given just above,
and $P(x,0)=P_0(x)$. Further, there exists a unique square matrix-valued
$g\in C^\infty\bigl([0,T];C^\infty(\Rb_-^{\times2}\times\Rb_+)\bigr)$
which satisfies the linear Fredholm equation,
\begin{equation*}
p(y+z+x;t)=g(y,z;x,t)-\int_{-\infty}^0 g(y,\xi;x,t)p(\xi+z+x,t)\,\mathrm{d}\xi,
\end{equation*}
corresponding to the operator equation $P=G(\id-P)$ for $G=G(x,t)$. 

Let us now briefly outline why the solution flow for $G$ given in 
Definition~\ref{def:linearsystem} is a Fredholm Grassmannian flow;
many more details can be found in Beck \textit{et al.\/ }~\cite{BDMSI,BDMSII},
Doikou \textit{et al.\/ }~\cite[Sec.~2.3]{DMS20}
and Doikou \textit{et al.\/ }~\cite{DMSW20}.
Detailed introductions and examples of Fredholm Grassmann manifolds
are given in Pressley and Segal~\cite{PS}, Abbondandolo and Majer~\cite{AMajer}
and Andruchow and Larotonda~\cite{AL}.
Suppose $\Hb$ is a given separable Hilbert space.
The Fredholm Grassmannian of all subspaces of $\Hb$ that are comparable
in size to a given closed subspace $\Vb\subset\Hb$ is defined as follows;  
see Segal and Wilson~\cite{SW}. 
\begin{definition}[Fredholm Grassmannian]\label{def:FredholmGrassmannian}
Let $\Hb$ be a separable Hilbert space with a given decomposition
$\Hb=\Vb\oplus\Vb^\perp$, where $\Vb$ and $\Vb^\perp$ are infinite
dimensional closed subspaces. The Grassmannian $\Gr(\Hb,\Vb)$
is the set of all subspaces $\Wb$ of $\Hb$ such that:
\begin{enumerate}
\item[(i)] The orthogonal projection $\mathrm{pr}\colon\Wb\to\Vb$ is 
a Fredholm operator, indeed it is a Hilbert--Schmidt perturbation
of the identity; and
\item[(ii)] The orthogonal projection $\mathrm{pr}\colon\Wb\to\Vb^\perp$ 
is a Hilbert--Schmidt operator.
\end{enumerate}
\end{definition}
We observe that as $\Hb$ is separable, any element in $\Hb$ has a representation
on a countable basis. Thus, for example, we can represent that element
via the sequence of coefficients which are associated with each of the basis elements.
Given an independent set of such sequences in $\Hb=\Vb\oplus\Vb^\perp$
which span $\Vb$, suppose we record them as columns in the infinite matrix
\begin{equation*}
W=\begin{pmatrix} \id+Q \\ P \end{pmatrix}.
\end{equation*}
In this representation, we suppose each column of $\id+Q$ lies in $\Vb$
and each column of $P$ lies in $\Vb^\perp$. Further, assume when we constructed $\id+Q$,
we ensured it was a Fredholm operator on $\Vb$ with $Q\in\mathfrak J_2(\Vb;\Vb)$,
where $\mathfrak J_2(\Vb;\Vb)$ is the class of Hilbert--Schmidt operators from $\Vb$ to $\Vb$,
equipped with the norm $\|Q\|_{\mathfrak J_2}^2\coloneqq\mathrm{tr}\,Q^\dag Q$ where `$\mathrm{tr}$'
is the trace operator. See Simon~\cite{Simon:Traces} for more details.
In addition, assume we constructed $P$ to ensure $P\in\mathfrak J_2(\Vb;\Vb^\perp)$,
the space of Hilbert--Schmidt operators from $\Vb$ to $\Vb^\perp$. 
Let $\Wb$ denote the subspace of $\Hb$
represented by the span of the columns of $W$. Let $\Vb_0\cong\Vb$ denote the canonical
subspace of $\Hb$ with the representation
\begin{equation*}
V_0=\begin{pmatrix} \id\\ O \end{pmatrix},
\end{equation*}
where $O$ is the infinite matrix of zeros. Assume that $\mathrm{det}_2(\id+Q)\neq0$.
Then the projections $\mathrm{pr}\colon\Wb\to\Vb_0$ and $\mathrm{pr}\colon\Wb\to\Vb_0^\perp$
respectively generate
\begin{equation*}
W^\parallel=\begin{pmatrix} \id+Q \\ O \end{pmatrix}
\quad\text{and}\quad
W^\perp=\begin{pmatrix} O \\ P \end{pmatrix}.
\end{equation*}
The subspace of $\Hb$ represented by the span of the columns of $W^\parallel$
naturally coincides with the subspace $\Vb_0$. Indeed, the transformation
$(\id+Q)^{-1}\in\mathrm{GL}(\Vb)$ transforms $W^\parallel$ to $V_0$.
Under this transformation the representation $W$ for $\Wb$ becomes
\begin{equation*}
\begin{pmatrix} \id \\ G \end{pmatrix},
\end{equation*}
where $G=P(\id+Q)^{-1}$. In fact, any subspace $\Wb$ that can
be projected onto $\Vb_0$ can be represented in this way and vice-versa.
In this representation the operators $G\in\mathfrak J_2(\Vb,\Vb^\perp)$
parametrise all the subspaces $\Wb$ that can be projected on to $\Vb_0$.
If $\mathrm{det}_2(\id+Q)=0$, then the projection above is not possible.
In this instance we simply choose a different representative coordinate chart/patch.
It is always possible to choose a subspace $\Vb_0^\prime\cong\Vb$ of $\Hb$
such that the projection $\Wb\to\Vb_0^\prime$ is an isomorphism;
see Pressley and Segal~\cite[Prop.~7.1.6]{PS}.
For further details on coordinate patches see Beck \textit{et al.}~\cite{BDMSI}
and Doikou \textit{et al.}~\cite{DMSW20} and Doikou \textit{et al.}~\cite{DMS20}.
In conclusion we observe that the flow for $G$ given in 
Definition~\ref{def:linearsystem}, taking $Q$ to be `$-P$'
and under the assumptions stated for $p_0$ directly after,
including the regularised determinant restriction,
represents a Fredholm Grassmannian flow.

%
%

\section{Hilbert--Schmidt operator and concatenation algebras}\label{sec:operatoralgebras} 
Let $\Cb\coloneqq\cup_{n\geqslant0}\Cb(n)$ denote the set of all compositions.
Herein we consider an abstract version of the operator algebra generated
by a parameter-dependent operator $P$, its derivatives with respect to the parameter,
and by $U\coloneqq(\id-P)^{-1}$ as well as the derivatives of $U$ with respect to the parameter.
With $P_n\coloneqq\pa^nP$ and $U_n\coloneqq\pa^nU$ for all non-negative integers $n$,
let $\mathbb P$ denote the alphabet $\{P_n\}_{n\geqslant0}$ and $\Ub$ denote the alphabet
$\{U_n\}_{n\geqslant0}$. Then we denote by $\Rb\la\Pb\cup\Ub\ra$,
the Hilbert--Schmidt operator algebra over $\Rb$, generated by $\Pb$ and $\Ub$.
This algebra is unital with the identity operator as the unit element.
One natural approach to generating an abstract version
of the algebra $\Rb\la\Pb\cup\Ub\ra$ is as follows.

Let $\underline{\mathbb N}$ denote the set of non-negative integers $\underl{\und{n}}$.
We distinguish such a set of integers for the following reason.
In the first stage of the abstraction we essentially replace the operators $P_n$ by integers $n$,
and the operators $U_n$ by the integers $\underl{\und{n}}$, and construct a concatenation algebra
of words generated by the alphabets $\Nb$ and $\underline{\Nb}$. 
We denote by $\Rb\la\Nb\cup\underline{\Nb}\ra$ the non-commutative concatenation polynomial
algebra over $\Rb$ generated words constructed from the alphabet $\Nb\cup\underline{\Nb}$;
see Reutenauer~\cite{Reu}. The neutral element is the empty word $\emptyset$,
so that the concatenation of any word from $\Rb\la\Nb\cup\underline{\Nb}\ra$ with $\emptyset$
and vice-versa, just generates the original word.
Hence $\Rb\la\Nb\cup\underline{\Nb}\ra$ is a unital algebra.
The neutral element $\emptyset$ plays the role equivalent
to that of the identity operator in $\Rb\la\Pb\cup\Ub\ra$.
Indeed, from their definitions we observe that $\Rb\la\Nb\cup\underline{\Nb}\ra$
and $\Rb\la\Pb\cup\Ub\ra$ are \emph{isomorphic}.
We can also define on $\Rb\la\Nb\cup\underline{\Nb}\ra$ a derivation operation $\tsd$ that
represents the abstraction of the derivative operator $\pa$ acting on the operators $P$ and $U$.
Indeed operator $\tsd\colon\Rb\la\Nb\cup\underline{\Nb}\ra\to\Rb\la\Nb\cup\underline{\Nb}\ra$
is the linear operator such that for any letter $n\in\Nb$ in $\Rb\la\Nb\cup\underline{\Nb}\ra$,
$\tsd\colon n\mapsto (n+1)$, and for any letter
$\und{n}\in\underline{\Nb}$ in $\Rb\la\Nb\cup\underline{\Nb}\ra$,
$\tsd\colon\underl{\und{n}}\mapsto\underl{\und{(n+1)}}$.
Furthermore for any word $a_1a_2a_3\cdots a_n\in\Rb\la\Nb\cup\underline{\Nb}\ra$,
where the letters $a_1$, $a_2$,\ldots,$a_n$ can be either from $\Nb$ or $\underline{\Nb}$,
the operator $\tsd$ satisfies the Leibniz derivation property:
\begin{equation*}
\tsd(a_1a_2a_3\cdots a_n)=\sum_{k=1}^{n}a_1a_2\cdots a_{k-1}(\tsd a_k)a_{k+1}\cdots a_n. 
\end{equation*}
The first crucial result of this section is that each of the letters in $\underline{\Nb}$
has a linear expansion in terms of monomials of the form
$\und{0}a_1\und{0}a_2\und{0}\cdots\und{0}a_n\und{0}$ with $a_1,a_2,\ldots,a_n\in\Nb$ \emph{only}.
Note here the letter $\und{0}$ corresponds to $U_0\equiv U$.
This is the abstract version of the statement of Theorem~\ref{thm:invOpII}. 
We state and prove this result after the following definition.
We denote by $\Nb^\ast$, the free monoid of words on $\Nb$.
\begin{definition}[Signature character]\label{def:signaturecharacter}
Let $a_1a_2\cdots a_n$ be a word from $\Nb^\ast$.
We associate with any such word the signature character $\chi\colon\Nb^\ast\to\Rb$,
given by
\begin{equation*}
\chi\bigl(a_1a_2\cdots a_n\bigr)
\coloneqq\begin{pmatrix} a_1+\cdots+a_n\\a_n\end{pmatrix}
\begin{pmatrix} a_{1}+a_{2}+\cdots+a_n\\a_{2}\end{pmatrix}
\cdots\begin{pmatrix} a_{n-1}+a_k\\a_{n-1}\end{pmatrix}\begin{pmatrix} a_n\\a_n\end{pmatrix},
\end{equation*}
where each of the factors shown on the right is a Leibniz coefficient, so that for example,
the penultimate factor is $a_{n-1}+a_n$ choose $a_{n-1}$.
\end{definition}
Let us introduce some notation. Given a word $w=a_1a_2\cdots a_n$
generated using letters $a_1$, $a_2$, \ldots, $a_n$ from $\Nb$,
let $w_{\und{0}}$ denote the corresponding word in $\Rb\la\Nb\cup\underline{\Nb}\ra$
of the form
\begin{equation*}w_{\und{0}}\coloneqq\und{0}a_1\und{0}a_2\und{0}\cdots\und{0}a_n\und{0},
\end{equation*}
i.e.\/ a single letter $\und{0}$ exists between each letter $a_1$, $a_2$, \ldots, $a_n$ from $\Nb$,
and at each end. 
\begin{theorem}[Signature operator expansion]\label{thm:underlinepolyexpansion}
Given any $n\in\Nb$, we have, in the algebra $\Rb\la\Nb\cup\underline{\Nb}\ra$,
for any corresponding integer letter $\und{n}$ from $\underline{\Nb}$ that:
\begin{equation*}
\und{n}=\sum_{w\in\Cb(n)}\chi(w)\cdot w_{\und{0}}.
\end{equation*}
\end{theorem}
\begin{proof}
We begin by recreating some of the basic results for
the operators $P$ and $U$ from $\Rb\la\Pb\cup\Ub\ra$ in Lemma~\ref{lemma:invOpI}
in terms of the corresponding elements from $\Rb\la\Nb\cup\underline{\Nb}\ra$.
With the letter $0$ corresponding to $P$, the letter $\und{0}$ corresponding
to $U$ and with $\emptyset$ representing the
neutral element corresponding to the identity $\id$ in $\Rb\la\Pb\cup\Ub\ra$,
we have
\begin{equation*}
\und{0}\coloneqq (\emptyset-0)^{-1} \qquad\text{and thus}\qquad
\und{0}\equiv\emptyset+0\und{0}\equiv\emptyset+\und{0}0.
\end{equation*}
These correspond to the first set of statements in Lemma~\ref{lemma:invOpI}.
Next using that $\tsd(\und{n})=\und{(n+1)}$, we apply $\tsd$ several times to
the identity $\und{0}\equiv\emptyset+0\und{0}$. Thus we have
\begin{equation*}
\und{1}\equiv\begin{pmatrix}1\\1\end{pmatrix}\cdot 1\und{0}+\begin{pmatrix}1\\0\end{pmatrix}\cdot 0\und{1}
\qquad\Leftrightarrow\qquad\und{1}\equiv\begin{pmatrix}1\\1\end{pmatrix}\cdot\und{0}1\und{0}.
\end{equation*}
This establishes the statement of the lemma for $\und{1}$. 
One further iteration of this procedure is insightful.
Applying $\tsd^2$ to the identity $\und{0}\equiv\emptyset+0\und{0}$, we observe
\begin{equation*}
\und{2}\equiv\begin{pmatrix}2\\2\end{pmatrix}\cdot2\und{0}+\begin{pmatrix}2\\1\end{pmatrix}\cdot1\und{1}
+\begin{pmatrix}2\\0\end{pmatrix}\cdot0\und{2}
\quad\Leftrightarrow\quad
\und{2}\equiv\begin{pmatrix}2\\2\end{pmatrix}\cdot\und{0}2\und{0}
+\begin{pmatrix}2\\1\end{pmatrix}\cdot\und{0}1\und{1}
\end{equation*}
Substituting for $\und{1}$ into the right-hand side, we find,
\begin{equation*}
\und{2}\equiv\begin{pmatrix}2\\2\end{pmatrix}\cdot\und{0}2\und{0}+\begin{pmatrix}2\\1\end{pmatrix}
\begin{pmatrix}1\\1\end{pmatrix}\cdot\und{0}1\und{0}1\und{0}.
\end{equation*}
This corresponds to the case $n=2$ in the statement of the lemma.
We now focus on the overall proof which proceeds by induction.
We have seen the statement of the lemma holds for $n=1$.
Assume the statement holds for all $1\leqslant k\leqslant n$.
We now show it holds for $(n+1)$, i.e.\/ that we have
\begin{equation*}
\und{(n+1)}=\sum_{w\in\Cb(n+1)}\chi(w)\cdot w_{\und{0}}.
\end{equation*}
We apply $\tsd^{n+1}$ to the identity $\und{0}\equiv\emptyset+0\und{0}$.
Using the Leibniz rule and combining the final term $0\und{(n+1)}$ on the right
with the term $\tsd^{n+1}\und{0}\equiv\und{(n+1)}$ on the left,
in the same manner as we did above, we obtain
\begin{align*}
\und{(n+1)}\equiv&
\begin{pmatrix}n+1\\n+1\end{pmatrix}\cdot\und{0}(n+1)\und{0}
+\begin{pmatrix}n+1\\n\end{pmatrix}\cdot\und{0}n\und{1}
+\begin{pmatrix}n+1\\n-1\end{pmatrix}\cdot\und{0}(n-1)\und{2}\\
&\;+\begin{pmatrix}n+1\\n-2\end{pmatrix}\cdot\und{0}(n-2)\und{3}
+\cdots+\begin{pmatrix}n+1\\2\end{pmatrix}\cdot\und{0}2\und{(n-1)}
+\begin{pmatrix}n+1\\1\end{pmatrix}\cdot\und{0}1\und{n}.
\end{align*}
We have assumed that for each $k$ such that $2\leqslant k\leqslant n$, we have
$\und{k}$ is the appropriate linear combination of all the compositions of $k$.
The first term on the right above involving $\und{0}(n+1)\und{0}$ is the only
composition of $n+1$ starting with the digit $n+1$---ignoring the $\und{0}$'s.
This term on the right above matches the term in $\und{(n+1)}$
corresponding to the single digit composition of $n+1$,
with the correct coefficient $\chi\bigl((n+1)\bigr)$.  
Now consider the next term on the right involving $\und{0}n\und{1}$.
If we substitute for $\und{1}$ we obtain the the only composition of $n+1$
starting with the digit $n$---again ignoring the $\und{0}$'s---namely
the term $\und{0}n\und{0}1\und{0}$. The coefficient is
\begin{equation*}
\begin{pmatrix}n+1\\n\end{pmatrix}\begin{pmatrix}1\\1\end{pmatrix}=\chi\bigl(n1\bigr).
\end{equation*}
Indeed, for each $1\leqslant k\leqslant n$, consider the term involving
$\und{0}(n+1-k)\und{k}$ on the right in the Leibniz expansion above.
We substitute for $\und{k}$ which is the appropriate linear combination of
all compositions of $k$. We observe this exhausts all the possible compositions
of $n+1$ that start with the digit $n+1-k$---ignoring the $\und{0}$'s.
Now consider the corresponding coefficients. Consider a generic term $w_{\und{0}}$ 
in $\und{k}$, for which the $\und{0}$-stripped version of $w_{\und{0}}$, namely $w$,
is a composition of $k$. The coefficient of $w_{\und{0}}$ is $\chi(w)$. 
Its coefficient in the Leibniz expansion above would be
\begin{equation*}
\begin{pmatrix}n+1\\n+1-k\end{pmatrix}\,\chi(w)=\chi\bigl((n+1-k)w\bigr).
\end{equation*}
This last result follows from Definition~\ref{def:signaturecharacter} for $\chi(w)$:
for any word $w$ which is a composition of $k$ and for which for a given letter $a$,
the word $aw$ is a composition of $n+1$, then $a=n+1-k$ and the result above holds.
This completes the proof.\qed
\end{proof}
The significance of Theorem~\ref{thm:underlinepolyexpansion} is that every element
$\und{n}\in\Rb\la\Nb\cup\underline{\Nb}\ra$ for $n\in\Nb$ can be expressed as a linear
combination of monomials of the form $\und{0}a_1\und{0}a_2\und{0}\cdots\und{0}a_n\und{0}$
with all the $a_k\in\Nb$, i.e.\/ from the integer alphabet only, not involving $\underline{\Nb}$.
This means that we can use monomials of this form as a basis for $\und{n}\in\Rb\la\Nb\cup\underline{\Nb}\ra$.
Indeed it thus is sufficient for us to consider the concatenation algebra $\und{n}\in\Rb\la\Cb_{\und{0}}\ra$
of monomials of the form $\und{0}a_1\und{0}a_2\und{0}\cdots\und{0}a_n\und{0}$ from 
$\Cb_{\und{0}}$, the set of all compositions with the letter $\und{0}$
squeezed between each digit as well as being added to each end.

\section{Hankel kernel and P\"oppe algebras}\label{sec:kernelalgebras}
We have seen that, as a result of Theorem~\ref{thm:underlinepolyexpansion},
it is sufficient for us to consider the concatenation algebra $\und{n}\in\Rb\la\Cb_{\und{0}}\ra$
of monomials of the form $\und{0}a_1\und{0}a_2\und{0}\cdots\und{0}a_n\und{0}$ where
$a_1a_2\cdots a_n\in\Cb$. This means instead of the Hilbert--Schmidt operator algebra
$\Rb\la\Pb\cup\Ub\ra$, it is sufficient for us to consider the Hilbert--Schmidt
operator algebra $\Rb\la\Pb_U\ra$ generated by operator monomials of the form
$UP_{a_1}UP_{a_2}U\cdots UP_{a_n}U$. 
Indeed we could have deduced this last statement
directly from Theorem~\ref{thm:invOpII}. Our goal herein is to consider
the algebra of Hilbert--Schmidt Hankel kernels generated by applying the 
the kernel bracket $[\,\cdot\,]$ to such operator monomials. 
To this end, let $\Rb\la[\mathbb P_U]\ra$ denote the non-commutative polynomial matrix algebra
generated by \emph{kernel monomials} of the form $[UP_{a_1}UP_{a_2}U\cdots UP_{a_n}U]$ 
where $a_1a_2\cdots a_n\in\Cb$. We call this the \emph{Hankel kernel algebra}.
We endow this algebra with a product equivalent to
the P\"oppe product in Lemma~\ref{lemma:kernelproductrule}
for such kernel monomials as follows.
\begin{lemma}[P\"oppe product for kernel monomials]\label{lemma:Poppeprodkernels}
Suppose $u=a_1\cdots a_k\in\Cb$ and $v=b_1\cdots b_\ell\in\Cb$ 
while $a,b\in\mathbb N$. Let $\mm_u$ and $\mm_v$ denote the respective kernel monomials
$\mm_u\coloneqq UP_{a_1}UP_{a_2}U\cdots UP_{a_k}U$ and
$\mm_v\coloneqq UP_{b_1}UP_{b_2}U\cdots UP_{b_\ell}U$.
Then we have 
\begin{equation*}
[\mm_uP_aU][UP_b\mm_v]
=[\mm_uP_{a+1}UP_b\mm_v]+[\mm_uP_aUP_{b+1}\mm_v]+2\,[\mm_uP_aUP_1UP_b\mm_v].
\end{equation*}
\end{lemma}
\begin{proof}
Using that $U\equiv\id+UP\equiv\id+PU$,
the P\"oppe product rule from Lemma~\ref{lemma:kernelproductrule}
and that $\pa U=U(\pa P)U$, we observe
\begin{align*}
[\mm_uP_aU][UP_b\mm_v]&=[\mm_uP_a+\mm_uP_aUP][PUP_b\mm_v+P_b\mm_v]\\
&=\;[\mm_uP_a][P_b\mm_v]+[\mm_uP_a][PUP_b\mm_v]\\
&\;+[\mm_uP_aUP][P_b\mm_v]+[\mm_uP_aUP][PUP_b\mm_v]\\
=&[\mm_u\pa(P_aP_b)\mm_v]+[\mm_u\pa(P_aP)UP_b\mm_v]\\
&\;+[\mm_uP_aU\pa(PP_b)\mm_v]+[\mm_uP_aU\pa(PP)UP_b\mm_v]\\
=&[\mm_u(\pa P_a)P_b\mm_v]+[\mm_uP_a(\pa P_b)\mm_v]\\
&\;+[\mm_u(\pa P_a)PUP_b\mm_v]+[\mm_uP_a(\pa P)UP_b\mm_v]\\
&\;+[\mm_uP_aU(\pa P)P_b\mm_v]+[\mm_uP_aUP(\pa P_b)\mm_v]\\
&\;+[\mm_uP_aU(\pa P)PUP_b\mm_v]+[\mm_uP_aUP(\pa P)UP_b\mm_v]\\
=&[\mm_u(\pa P_a)UP_b\mm_v]+[\mm_uP_aU(\pa P_b)\mm_v]\\
&\;+[\mm_uP_aU(\pa P)UP_b\mm_v]+[\mm_uP_aU(\pa P)UP_b\mm_v]\\
=&[\mm_u(\pa P_a)UP_b\mm_v]+[\mm_uP_aU(\pa P_b)\mm_v]\\
&\;+2\,[\mm_uP_aU(\pa P)UP_b\mm_v],
\end{align*}
which establishes the result.\qed
\end{proof}
\begin{remark}
Recall from Lemma~\ref{lemma:kernelproductrule}, we implicitly interpret multiple
kernel products of the form $[\,\cdot\,][\,\cdot\,]\cdots[\,\cdot\,][\,\cdot\,]$
to mean $[\,\cdot\,](y,0;x)[\,\cdot\,](0,0;x)\cdots[\,\cdot\,](0,0;x)[\,\cdot\,](0,z;x)$.
\end{remark}
\begin{definition}[Signature kernel expansions]\label{def:comppolyexpnker}
Motivated by Theorems~\ref{thm:invOpII} and \ref{thm:underlinepolyexpansion},
we identify the following specific linear \emph{signature kernel expansions},
which for any $n\in\Nb$, with the sum over all compositions $a_1a_2\cdots a_k\in\Cb(n)$,
are given by 
\begin{equation*}
[U_n]=\sum\chi\bigl(a_1a_2\cdots a_k\bigr)\,
[UP_{a_1}UP_{a_2}U\cdots UP_{a_{k-1}}UP_{a_k}U].
\end{equation*}
\end{definition}
Our goal now is to generate an abstract version of the Hankel kernel algebra
$\Rb\la[\mathbb P_U]\ra$. One natural approach is as follows. 
Let $\Rb\la\Cb\ra_\ast$ denote the non-commutative 
polynomial algebra over $\Rb$ generated by composition elements from $\Cb$, 
endowed with the following \emph{P\"oppe product} for compositions.
\begin{definition}[P\"oppe product for compositions]\label{def:Poppeproductcompositions}
Consider two compositions $ua$ and $bv$ in $\Cb$, where we distinguish the last and first letters
in the former and latter compositions as the natural numbers $a$ and $b$ respectively.
We define the \emph{P\"oppe product} `$\ast$' from 
$\Rb\la\Cb\ra_\ast\times\Rb\la\Cb\ra_\ast$ to $\Rb\la\Cb\ra_\ast$
for the compositions $ua$ and $bv$ to be, 
\begin{equation*}
(ua)\ast(bv)=u(a+1)bv+ua(b+1)v+2\cdot(ua1bv).
\end{equation*}
\end{definition}
We see some alternative prescriptions for this product in Section~\ref{sec:Laxhierarchy};
see Lemma~\ref{lemma:Poppeproductandderivation}.
\begin{remark}[Stripped or $\und{0}$-padded compositions?]
One might wonder why we have not considered here
the non-commutative polynomial algebra over $\Rb$ generated by composition elements
from $\Cb_{\und{0}}$, where $\Cb_{\und{0}}$ is the set of composition monomials
of the form $w_{\und{0}}=\und{0}a_1\und{0}a_2\und{0}\cdots\und{0}a_n\und{0}$
where $w=a_1a_2\cdots a_n\in\Cb$. Indeed this would be natural.
We would define the P\"oppe product for compositions from $\Cb_{\und{0}}$
in exactly the same manner as in Definition~\ref{def:Poppeproductcompositions}
above. So for example $ua\in\Cb$ is replaced by $(ua)_{\und{0}}$,
i.e.\/ a single $\und{0}$ is squeezed between all the letters in $ua$
as well as on either end. All the terms in the product $(ua)\ast(bv)$
can be interpretted analogously. In this $\und{0}$-padded abstraction,
we interpret $\und{0}$ as corresponding to the \emph{pseudo-passive} operator `$U$'
in the kernel monomials $[UP_{a_1}UP_{a_2}U\cdots UP_{a_n}U]$.
With this interpretation there is a one-to-one correspondence
between the monomials $\und{0}a_1\und{0}a_2\und{0}\cdots\und{0}a_n\und{0}$
and $[UP_{a_1}UP_{a_2}U\cdots UP_{a_n}U]$. The P\"oppe products given 
in Lemma~\ref{lemma:Poppeprodkernels} and Definition~\ref{def:Poppeproductcompositions}
are mirror images of each other and the algebras generated by
the corresponding sets of monomials are seen to be isomorphic.
However in practice, keeping the pseudo-passive $\und{0}$'s is cumbersome,
and not required.
\end{remark}
Herein we use the stripped abstraction of $\Rb\la[\mathbb P_U]\ra$,
replacing it by its isomorphic equivalent $\Rb\la\Cb\ra_\ast$. 
The coding we use is to replace the kernel monomials $[UP_{a_1}U\cdots UP_{a_n}U]$
by the $\und{0}$-stripped compositions $a_1\cdots a_n$ and use the
P\"oppe product for compositions in Definition~\ref{def:Poppeproductcompositions}
in place of the P\"oppe product in Lemma~\ref{lemma:Poppeprodkernels}.
Hence, to emphasise, $\Rb\la\Cb\ra_\ast$ is the non-commutative polynomial
algebra over $\Rb$ generated by the set of compositions $\Cb$ with
the product being the P\"oppe product in Definition~\ref{def:Poppeproductcompositions}.
Further, from their definitions, we observe that
$\Rb\la\Cb\ra_\ast$ and $\Rb\la[\mathbb P_U]\ra$ are isomorphic.
\begin{remark}\label{rmk:unitalemptycomp}
Note we do not require $\Rb\la[\mathbb P_U]\ra$
and $\Rb\la\Cb\ra_\ast$ to be unital algebras here nor
in the proof of our main result Theorem~\ref{thm:mainresult}
in Section~\ref{sec:existenceanduniqueness}. However
when we construct the signature co-algebra in Section~\ref{sec:signaturecoalgebra}
it is useful to introduce the empty composition $\nu\in\Rb\la\Cb\ra_\ast$
with the property $\nu\ast w=w\ast\nu=w$ for any composition $w$.
\end{remark}
The abstract versions of the signature expansions $[U_n]$ for kernels
in $\Rb\la[\mathbb P_U]\ra$ in Definition~\ref{def:comppolyexpnker}
are the following signature expansions in $\Rb\la\Cb\ra_\ast$.
\begin{definition}[Signature expansions]\label{def:polycomp}
For any $n\in\Nb$ we define the following linear
\emph{signature expansions} $\und{n}\in\Rb\la\Cb\ra_\ast$:
\begin{equation*}
\wdh{\und{n}}\coloneqq\sum_{w\in\Cb(n)}\chi(w)\cdot w.
\end{equation*}
\end{definition}
\begin{remark}[Notation]
Note, in Definition~\ref{def:polycomp} and hereafter, we use $\und{1}$, $\und{2}$, \ldots, $\und{n}$  
and so forth to denote linear signature expansions in $\Rb\la\Cb\ra_\ast$ associated with the kernel
algebra $\Rb\la[\mathbb P_U]\ra$, i.e.\/ as linear expansions in terms of composition monomials
in $\Rb\la\Cb\ra_\ast$. Hitherto, see for example Theorem~\ref{thm:underlinepolyexpansion},
we used $\und{n}$ to denote signature expansions in $\Rb\la\Nb\cup\underline{\Nb}\ra$
associated with the operator algebra $\Rb\la\Pb\cup\Ub\ra$. 
\end{remark}
\begin{remark}[Convention]
In Definition~\ref{def:polycomp} we observe a notation convention we have used
thus far for expansions in $\Rb\la\Cb\ra_\ast$ and $\Rb\la\Cb_{\und{0}}\ra$ and
which we use hereafter.
That is, we distinguish between the real coefficients associated with
combinatorial elements such as $a_1\cdots a_k\in\Cb$ by separating them
by a `$\;\cdot\;$', with the real-valued coefficients residing on the
left and the compositions residing on the right. So for example if we
write $3\cdot221$, this is to interpreted as $3$ copies of the composition element `$221$'. 
\end{remark}
As for the abstract operator algebra in Section~\ref{sec:operatoralgebras},
we can define on $\Rb\la\Cb\ra_\ast$ a derivation operation $\md$ as follows.
\begin{definition}[Derivation in composition kernel algebra]\label{def:derivation}
Given any monomial $w\in\Rb\la\Cb\ra_\ast$ with $w=a_1a_2\cdots a_n$ we define
the \emph{derivation} $\md$ of the composition $w$ to be the linear expansion:   
\begin{equation*}
  \md w=\sum_{k=1}^na_1a_2\cdots a_{k-1}(a_k+1)a_{k+1}\cdots a_n
  +\sum_{k=1}^{n+1}a_1a_2\cdots a_{k-1}\,1\,a_k\cdots a_n.
\end{equation*}
In the second sum, the $k=1$ and $k=n+1$ cases
correspond to a `$1$' being appended, respectively, to the front and
then end of the composition $a_1a_2\cdots a_n$.
\end{definition}
\begin{example}
For any two arbitrary single letter compositions
$a,b\in\Rb\la\Cb\ra_\ast$ we have $\md(ab)=1ab+(a+1)b+a1b+a(b+1)+ab1$.
More generally, for any two arbitrary compositions $ua,bv\in\Rb\la\Cb\ra_\ast$ where we 
distinguish the final and beginning single letters $a$ and $b$ as shown,  
we have $\md(uabv)=(\md u)abv+u\bigl(\md(ab)\bigr)v+uab(\md v)$.
\end{example}
\begin{remark}
This can be viewed as the result of the derivation of the $\und{0}$-padded
version $w_{\und{0}}$ of $w$ using the Leibniz rule. In other words
the definition above is just the abstract encoding of applying
the derivation operation $\pa$ to $[UP_{a_1}UP_{a_2}U\cdots UP_{a_n}U]$.
Note the linear operations $\pa$ and $[\,\cdot\,]$ commute.
\end{remark}
Just as for the abstract operator algebra, for the abstract kernel algebra
here we have the following property.
\begin{lemma}
For any signature expansion $\und{n}\in\Rb\la\Cb\ra_\ast$,
the action of the derivation operation is $\md\colon\wdh{\und{n}}\mapsto\wdh{\und{(n+1)}}$,
i.e.\/ we have $\md(\wdh{\und{n}})=\wdh{\und{(n+1)}}$.
\end{lemma}
\begin{proof}
The derivation $\md(\wdh{\und{n}})$ corresponds to the signature 
expansion in Definition~\ref{def:polycomp} with all the compositions $w\in\Cb(n)$
replaced by $\md (w)$. Using that $\Rb\la\Cb\ra_\ast$ and $\Rb\la[\mathbb P_U]\ra$
are isomorphic, this corresponds to the signature expansion for kernels
in Definition~\ref{def:comppolyexpnker} with each term
$[UP_{a_1}UP_{a_2}U\cdots UP_{a_k}U]$
therein replaced by $\pa[UP_{a_1}UP_{a_2}U\cdots UP_{a_k}U]$.
We know this expansion, by appling the differential operator $\pa$
and the kernel bracket operater $[\,\cdot\,]$ to the result of Theorem~\ref{thm:invOpII},
corresponds to $\pa[U_n]$ which equals $[U_{n+1}]$. The signature kernel expansion for
$[U_{n+1}]$ corresponds to that for $\wdh{\und{(n+1)}}$,
by comparing term for term, their signature expansions in terms of compositions
$w\in\Cb(n+1)$.\qed
\end{proof}

\section{Signature co-algebra}\label{sec:signaturecoalgebra}
Herein we motivate and introduce the signature co-algebra $\Sb$. 
Though not strictly necessary to prove our main result, our impetus for
constructing the co-algebra will be useful in 
helping to elucidate our main result in Section~\ref{sec:existenceanduniqueness}. 
To motivate our construction of the signature co-algebra consider the 
product of two arbitrary signature expansions $\wdh{\und{m}}$ and $\wdh{\und{n}}$
in $\Rb\la\Cb\ra_\ast$:
\begin{equation*}
\wdh{\und{m}}\ast\wdh{\und{n}}=\sum\chi(u)\chi(v)\cdot u\ast v,
\end{equation*}
where the sum is over all $u\in\Cb(m)$ and $v\in\Cb(n)$, 
and the product is naturally the P\"oppe product for compositions.
For the real-valued coefficients $\chi(u)$ and $\chi(v)$ we have
$\chi(u)\chi(v)\equiv\chi(v)\chi(u)$
for any pair of compositions $u,v\in\Cb$; recall the signature character map $\chi$
from Definition~\ref{def:signaturecharacter}.
In the product above, we can expand the products $u\ast v$ into linear
combinations of compositions using Definition~\ref{def:Poppeproductcompositions}
for the P\"oppe product. However we observe that, in the sum on the
right above for the product $\wdh{\und{m}}\ast\wdh{\und{n}}$, once we have expanded
all the terms $u\ast v$ for all $u\in\Cb(m)$ and $v\in\Cb(n)$, 
it might be useful to explicitly preserve the generators of specific compositions that appear
in the sum on the right; and we should extend this idea to higher degree
versions of such products. 
\begin{example}\label{ex:21prod}
Consider the linear combination
$c_{21}\cdot\wdh{\und{2}}\ast\wdh{\und{1}}+c_{12}\cdot\wdh{\und{1}}\ast\wdh{\und{2}}$
for some real constants $c_{21}$ and $c_{12}$.
The product $\wdh{\und{2}}\ast\wdh{\und{1}}$ equals 
\begin{align*}
\wdh{\und{2}}\ast\wdh{\und{1}}=&\;(\chi(2)\cdot2+\chi(11)\cdot 11)\ast(\chi(1)\cdot1)\\
=&\;\chi(2)\chi(1)\cdot2\ast1+\chi(11)\chi(1)\cdot 11\ast1\\
=&\;\chi(2)\chi(1)\cdot(31+22+2\cdot 211)+\chi(11)\chi(1)\cdot(121+112+2\cdot 1111)\\
=&\;\chi(2)\chi(1)\cdot(31+22)
+\chi(2)\chi(1)\chi(\mfq)\cdot(211)+\chi(11)\chi(1)\cdot(121+112)\\
&\;+\chi(11)\chi(1)\chi(\mfq)\cdot(1111),
\intertext{and so we can deduce the product $\wdh{\und{1}}\ast\wdh{\und{2}}$ equals}
\wdh{\und{1}}\ast\wdh{\und{2}}=&\;
\chi(1)\chi(2)\cdot(13+22)
+\chi(1)\chi(2)\chi(\mfq)\cdot(112)+\chi(1)\chi(11)\cdot(121+211)\\
&\;+\chi(1)\chi(11)\chi(\mfq)\cdot(1111).
\end{align*}
Note we distinguish a special element `$\mfq$' to indicate
it results from the third term in the P\"oppe product on compositions, so $\chi(\mfq)=2$.
For this example, think of the symbol `$\chi(\mfq)$' simply as a proxy for the real coefficient `$2$'.
So if we expand the linear combination
$c_{21}\cdot\wdh{\und{2}}\ast\wdh{\und{1}}+c_{12}\cdot\wdh{\und{1}}\ast\wdh{\und{2}}$,
by preserving the signature character coefficients in the manner above,
we know the composition $31$ was the result of the product $\wdh{\und{2}}\ast\wdh{\und{1}}$,
because the coefficient of $31$ is $\chi(2)\chi(1)$.
On the other hand, the composition $211$ which would appear with the coefficient
$c_{21}\chi(2)\chi(1)\chi(\mfq)+c_{12}\chi(1)\chi(11)$ was generated as follows.
It was generated by the third term in the product $2\ast1$ coming from the product
$\wdh{\und{2}}\ast\wdh{\und{1}}$, i.e.\/ where a `$1$' is inserted between the $2$ and the $1$.
That it was this operation that generated this term is indicated by the product
of factors $\chi(2)\chi(1)\chi(\mfq)$,
i.e.\/ in particular that it involves the factor $\chi(\mfq)$.
It was also generated in the product $1\ast(11)$
coming from the product $\wdh{\und{1}}\ast\wdh{\und{2}}$, as indicated by the product of
factors $\chi(1)\chi(11)$. In other words, each composition can be associated
with a finite product of `signatures'.
\end{example}
This example illustrates the usefulness of carefully retaining the coefficients
when expanding linear combinations of monomials consisting of P\"oppe products of
signature expansions. More generally we require the multi-factor version
of the product of signature expansions given above of the form:
\begin{equation*}
\wdh{\und{n}}_1\ast\wdh{\und{n}}_2\ast\cdots\ast\wdh{\und{n}}_k
=\sum\chi(u_1)\chi(u_2)\cdots\chi(u_n)\cdot u_1\ast u_2\ast\cdots\ast u_k,
\end{equation*}
where the sum is over all $u_1\in\Cb(n_1)$, \ldots, $u_k\in\Cb(n_k)$.
Naturally we observe that any composition $w\in\Cb$ can result from
one or more P\"oppe products of the form $u_1\ast u_2\ast\cdots\ast u_k$.
It is now apparent that it is convenient to encode $\chi(u_1)\chi(u_2)\cdots\chi(u_k)$
as $\chi(u_1\smb u_2\smb\cdots\smb u_k)$. Indeed we can assume $\chi$ to act \emph{homomorpically}
on any such tensor product of compositions so that indeed
$\chi(u_1\smb u_2\smb\cdots\smb u_k)\equiv\chi(u_1)\chi(u_2)\cdots\chi(u_k)$.
Thus, for example, we now write:
\begin{align*}
\wdh{\und{2}}\ast\wdh{\und{1}}=&\;\chi(2\smb1)\cdot(31+22)
+\chi(2\smb\mfq\smb1)\cdot(211)+\chi(11\smb1)\cdot(121+112)\\
&\;+\chi(11\smb\mfq\smb1)\cdot(1111).
\end{align*}
The reason for introducing the special element `$\mfq$' is now explained.
When we encode the coefficients in this tensorially manner we need a mechanism
to record those compositions that are generated by the third term
in the P\"oppe product of two appropriate compositions.
We use this encoding explicitly when proving our main results
in Section~\ref{sec:existenceanduniqueness}, always keeping in mind
though that expressions of the form $\chi(u_1\smb u_2\smb\cdots\smb u_k)$,
where one of the tensorially elements may be `$\mfq$', are really to
interpreted as $\chi(u_1)\chi(u_2)\cdots\chi(u_k)$ and $\chi(\mfq)$
just acts as a proxy for the coefficient `$2$'.

The other important idea we extract from our computations above
is that it might be useful, for any given composition $w\in\Cb$, to determine
all the possible pairs of compositions that might have generated it. 
For example, we observe in Example~\ref{ex:21prod} above that
the composition $22$ can be generate through $2\ast1$ and also $1\ast2$.
On the other hand, the composition $112$ can be generated via
$11\ast1$, or via $1\ast2$ through the third term in the P\"oppe product. 
This can be used as a simple check that we have collated terms correctly,
see for example the rows in Table~\ref{table:KdV5}
in Section~\ref{sec:existenceanduniqueness}.
We explore this idea more formally just below. 
However, the notions we have discussed thusfar are all that we require from this section
for our proofs in Sections~\ref{sec:existenceanduniqueness} and \ref{sec:Laxhierarchy}.
At this point, the reader only focused on our main results can safely skip to 
Section~\ref{sec:existenceanduniqueness}.
\begin{remark}[Alternative encodings]
The encoding above involving the special element `$\mfq$'
is one functional encoding for our main results in Section~\ref{sec:existenceanduniqueness}.
However, it does not preserve the number of tensorial components.
It is perhaps more natural for the number of tensorial components to
match the number of P\"oppe products associated with generating
that particular composition. See Definiton~\ref{def:dePoppecoprod}
and in particular Lemma~\ref{lemma:co-prod} for de-P\"oppe co-product
just below to see how this can be achieved.
\end{remark}
The rest of this section is devoted to establishing a new co-algebra
we call the signature co-algebra. The motivating idea, from the notions just discussed,
is to formalise the process of determining all the possible P\"oppe products of compositions
that generate a given composition.
We begin by defining a co-product on $\Cb$, the de-P\"oppe co-product. 
\begin{definition}[De-P\"oppe co-product]\label{def:dePoppecoprod}
For any composition $w\in\Cb$, we define the \emph{de-P\"oppe co-product}
$\Delta(w)$ of $w$ to be   
\begin{equation*}
\Delta(w)=\sum_{u\otimes v\in\Cb^{\otimes2}}\la u\ast v, w\ra_{\Cb}\,u\otimes v,
\end{equation*}
where $\la\,\cdot\,,\,\cdot\,\ra_{\Cb}$
is the inner product on $\Cb$ defined for any $u,v\in\Cb$ by
\begin{equation*}
\la u, v\ra_{\Cb}=\begin{cases} 1 &\quad\text{if}~u=v,\\ 0 &\quad\text{if}~u\neq v,\end{cases}.
\end{equation*}
\end{definition}
We give the definition of the signature co-algebra first, and then prove that the
signature co-algebra is indeed a co-algebra second.
\begin{definition}[Signature co-algebra]\label{def:signaturecoalgebra}
We define the \emph{signature co-algebra} $\Sb$ as the co-algebra
$\Rb\la\Cb^\otimes\ra$ over $\Rb$ constructed from all possible
monomials $s_1\otimes s_2\otimes\cdots\otimes s_n$ chosen from $\Cb^{\otimes n}$
for all $n\in\Nb$. Here $\Cb^\otimes$ denotes $\cup_{n\geqslant0}\Cb^{\otimes n}$.
We define the co-product $\Delta\colon\Sb\to\Sb\otimes\Sb$ on $\Sb$ to be the
\emph{de-P\"oppe co-product} $\Delta$ in Definition~\ref{def:dePoppecoprod}.
We see from Lemma~\ref{lemma:co-prod} just below, $\Delta$ tensorally decomposes
any composition $w$ into the sum of all possible composition pairs that
produce $w$ via the P\"oppe product `$\ast$', including the 
empty composition $\nu$; see Remark~\ref{rmk:unitalemptycomp}.
The co-unit on $\Sb$ which we denote $\varepsilon\colon\Sb\to\Rb$ is given for any $w\in\Cb$ by 
\begin{equation*}
\varepsilon(w)\coloneqq\begin{cases}1\, &\quad\text{if}~w=\nu,\\ 0\, &\quad\text{if}~w\neq\nu,\end{cases}.
\end{equation*}
\end{definition}
We can derive an explicit formula for the de-P\"oppe product as follows.
To achieve this, setting $\boldsymbol 0\coloneqq\md^{-1}(1)$,
the following map $\theta\colon\Cb\cup\{\boldsymbol 0\}\to\Cb$
proves useful. For any composition $w\in\Cb$, we define:
\begin{equation*}
\theta\colon w\mapsto w;\quad\theta\colon \boldsymbol 0w\mapsto w;
\quad\theta\colon w\boldsymbol 0\mapsto w;
\quad\theta\colon\nu\mapsto\nu
\quad\text{and}\quad \theta\colon\boldsymbol 0\mapsto0\cdot\nu.
\end{equation*}
The coefficient in the image in the final case is $0\in\Rb$, so ultimately the term is zero.
\begin{lemma}[De-P\"oppe co-product formula]\label{lemma:co-prod}
The co-product $\Delta$ can be characterised as follows.
For any composition $a_1a_2\cdots a_n$ we have $\Delta(a_1a_2\cdots a_n)$ is given by:
\begin{multline*}
(\theta\otimes\theta)\circ
\Biggl(\sum_{k=0}^{n}\bigl(a_1\cdots a_{k-1}(a_k-1)\otimes a_{k+1}\cdots a_n
+a_1\cdots a_{k}\otimes(a_{k+1}-1)a_{k+2}\cdots a_n\bigr)\Biggr),
\end{multline*}
where the $k=0$ and $k=n$ terms are $\nu\otimes a_1a_2\cdots a_n$
and $a_1a_2\cdots a_n\otimes\nu$, respectively.
\end{lemma}
\begin{proof}
For any composition $a_1a_2\cdots a_n\in\Cb$, consider the argument of the map
$\theta\otimes\theta$ given in the statement of the theorem. This is equivalent to
the re-written form:
\begin{multline*}
\nu\otimes a_1a_2\cdots a_n+(a_1-1)\otimes a_2\cdots a_n
+a_1\otimes (a_2-1)a_3\cdots a_n\\
+a_1(a_2-1)\otimes a_3\cdots a_n
+a_1a_2\otimes(a_3-1)a_4\cdots a_n+\cdots
+a_1\cdots (a_{n-1}-1)\otimes a_n\\
+a_1\cdots a_{n-1}\otimes (a_n-1)+a_1a_2\cdots a_n\otimes\nu.
\end{multline*}
The first and last terms account for the possibility $w=a_1a_2\cdots a_n$ 
could be generated by $\nu\ast w$ and $w\ast\nu$, respectively.
In general we observe that for $k=1,\ldots,n-1$
the term $a_1\cdots a_ka_{k+1}\cdots a_n$ could be generated by either
$a_1\cdots a_{k-1}(a_k-1)\ast a_{k+1}\cdots a_n$ or
$a_1\cdots a_{k}\ast(a_{k+1}-1)a_{k+2}\cdots a_n$, explaining the
forms shown above. However for these terms care must be taken when
any one of $a_1$ through to $a_k$ equals `$1$'.
This is where the map $\theta\otimes\theta$ comes into play.
In such cases we symbolically have $\boldsymbol 0=a_k-1=\md^{-1}(a_k)\not\in\Cb$.
If $a_1=1$ or $a_n=1$, then the second and penultimate terms in the re-written
form above should not be present. If $a_2=1$ then the third and fourth terms
in the re-written form above should each collapse to $a_1\otimes a_3\cdots a_n$.
We apply a similar procedure if $a_3=1$ or $a_4=1$ and so forth up until the
case $a_{n-1}=1$. The action of $\theta\otimes\theta$ 
precisely enforces these collapses in the special cases mentioned.\qed
\end{proof}
\begin{example}
Consider the following examples illustrating 
the co-product $\Delta$ applied to some compositions.
Naturally for any integer $n\in\Cb$ we have $\Delta(n)=\emptyset$,
while for other compositions we have, for example,
$\Delta(21)=\nu\otimes21+1\otimes 1+21\otimes\nu$
and also $\Delta(32)=\nu\otimes32+2\otimes2+3\otimes 1+32\otimes\nu$, while:
\begin{align*}
\Delta(111)&=\nu\otimes111+1\otimes1+111\otimes\nu,\\
\Delta(121)&=\nu\otimes121+1\otimes 11+11\otimes1+121\otimes\nu,\\
\Delta(112)&=\nu\otimes112+1\otimes2+11\otimes 1+112\otimes\nu,\\
\Delta(241)&=\nu\otimes241+1\otimes 41+2\otimes31+23\otimes1+241\otimes\nu.
\end{align*}
\end{example}
The following result establishes that the signature co-algebra is indeed
a co-algebra; the proof is given in Appendix~\ref{sec:co-algebraproof}.

\begin{theorem}[Co-algebra structure]\label{thm:sigcoalg}
The signature co-algebra $\Sb$ is a co-algebra, in particular 
the co-unit $\varepsilon$ and co-product $\Delta$
satisfy the following defining axioms on $\Sb$:
(i) $(\id\otimes\varepsilon)\circ\Delta=(\varepsilon\otimes\id)\circ\Delta$ and
(ii) $(\id\otimes\Delta)\circ\Delta=(\Delta\otimes\id)\circ\Delta$.
\end{theorem}
\begin{remark}[The signature character map as a homomorphism]
With a slight abuse of notation we denote $\chi(\Sb)\cong\Rb$
as the space of corresponding signature character values.
In this context, we suppose the character map $\chi$ is
a homomorphism so that $\chi(u\otimes v)=\chi(u)\chi(v)$ for any pair $u,v\in\Sb$.
\end{remark}

\section{Hierarchy existence and uniqueness}\label{sec:existenceanduniqueness}
In the P\"oppe algebra $\Rb\la\Cb\ra_\ast$ we introduced in Definition~\ref{def:polycomp}
the linear signature expansions $\wdh{\und{n}}\in\Rb\la\Cb\ra_\ast$ given by:
\begin{equation*}
\wdh{\und{n}}\coloneqq\sum_{w\in\Cb(n)}\chi(w)\cdot w.
\end{equation*}
With the formal set up constructed in Sections~\ref{sec:kernelalgebras}
and \ref{sec:signaturecoalgebra}, we now tackle the following challenge.
We have defined the signature expansions $\wdh{\und{n}}$ in $\Rb\la\Cb\ra_\ast$ as
linear expansions in monomials of the form $w\in\Rb\la\Cb\ra_\ast$.
The question is, can we find polynomial expansions for the
basic monomial single letters $n\in\Rb\la\Cb\ra_\ast$, 
in terms of monomials of the form $\wdh{\und{a}}_1\ast\wdh{\und{a}}_2\ast\cdots\ast\wdh{\und{a}}_k$?
A positive answer guarantees the \emph{integrability} of the corresponding equation
of order $n$ in the non-commutative potential Korteweg--de Vries hierarchy.
\begin{example}[Korteweg--de Vries integrability]\label{ex:kdV}
From their definition, the first three signature expansions are $\wdh{\und{1}}=\chi(1)\cdot 1$,
$\wdh{\und{2}}=\chi(2)\cdot 2+\chi(11)\cdot 11$ and then
\begin{equation*}
\wdh{\und{3}}=\chi(3)\cdot 3+\chi(21)\cdot 21+\chi(12)\cdot12+\chi(111)\cdot111
~\Leftrightarrow~
3=\wdh{\und{3}}-3\cdot(21+12)-6\cdot111.
\end{equation*}
However we also observe that
$\wdh{\und{1}}\ast\wdh{\und{1}}=\bigl(\chi(1)\cdot 1\bigr)\ast\bigl(\chi(1)\cdot 1\bigr)=\chi(1)\chi(1)\cdot(1\ast1)$
and so
\begin{equation*}
\wdh{\und{1}}\ast\wdh{\und{1}}
=\chi(1)\chi(1)\cdot\bigl(21+12+\chi(\mfq)\cdot(111)\bigr)
=21+12+\chi(\mfq)\cdot 111.
\end{equation*}
Substituting the latter result into the former we observe
\begin{equation*}
3=\wdh{\und{3}}-c_{11}\cdot\bigl(\wdh{\und{1}}\ast\wdh{\und{1}}\bigr),
\end{equation*}
where $c_{11}=3$. This demonstrates that indeed $3\in\Rb\la\Cb\ra_\ast$ has a polynomial
expansion in terms of the signature expansions. If we translate this
statement back into the Hankel kernel algebra $\Rb\la[\mathbb P_U]\ra$
it becomes $[UP_3U]=[U_3]-3[U_1][U_1]$. This statement establishes that
the potential Korteweg--de Vries equations as an integrable Grassmannian flow.
\end{example}
\begin{example}[Korteweg--de Vries integrability: fifth order]\label{ex:fifthorder}
As in the last example, the question we need to answer is,
can the single monomial letter $5\in\Rb\la\Cb\ra_\ast$ be expressed as a
linear combination of monomials of the form $\wdh{\und{a}}_1\ast\wdh{\und{a}}_2\ast\cdots\ast\wdh{\und{a}}_k$?
A proposed polynomial of such signature expansions is
$\pi_5=\pi_5\bigl(\wdh{\und{1}},\wdh{\und{2}},\wdh{\und{3}},\wdh{\und{5}}\bigr)$, where 
\begin{equation*}
\pi_5\coloneqq c_5\cdot\wdh{\und{5}}+c_{31}\cdot\wdh{\und{3}}\ast\wdh{\und{1}}
+c_{22}\cdot\wdh{\und{2}}\ast\wdh{\und{2}}+c_{13}\cdot\wdh{\und{1}}\ast\wdh{\und{3}}
+c_{111}\cdot\wdh{\und{1}}\ast\wdh{\und{1}}\ast\wdh{\und{1}},
\end{equation*}
where $c_5$, $c_{31}$, $c_{22}$, $c_{13}$ and $c_{111}$ are constants.
So the question is, in $\Rb\la\Cb\ra_\ast$, can we find find values
for these coefficients $c_5$, $c_{31}$, $c_{22}$, $c_{13}$ and $c_{111}$ such that $5=\pi_5$?
We observe that, given the leading term in $\pi_5$, we should anticipate that $c_5=1$. 
The remaining combination of the quadratic terms
$\wdh{\und{3}}\ast\wdh{\und{1}}$, $\wdh{\und{2}}\ast\wdh{\und{2}}$, $\wdh{\und{1}}\ast\wdh{\und{3}}$
and the cubic term $\wdh{\und{1}}\ast\wdh{\und{1}}\ast\wdh{\und{1}}$ are chosen because these are the only
monomials of this form which when we substitute for the corresponding signature expansions
for $\wdh{\und{1}}$, $\wdh{\und{2}}$ and $\wdh{\und{3}}$ into them, and compute the P\"oppe product
of the resulting expansions, we will get compositions in $\Cb(5)$. And the compositions of $\Cb(5)$
represent the basis in $\Rb\la\Cb\ra_\ast$ we should use, guided by the linear signature expansion
for $\und{5}$ which we can write in the form $5=\und{5}-\text{`lower multi-part compositions'}$.
For example, the quadratic term $\wdh{\und{3}}\ast\wdh{\und{1}}$ equals
\begin{align*}
\wdh{\und{3}}\ast\wdh{\und{1}}
=&\;\bigl(\chi(3)\cdot3+\chi(21)\cdot21+\chi(12)\cdot12+\chi(111)\cdot111\bigr)\ast\bigl(\chi(1)\cdot1\bigr)\\
=&\;\chi(3\smb1)\cdot3\ast1+\chi(21\smb1)\cdot21\ast1+\chi(12\smb1)\cdot12\ast1+\chi(111\smb1)\cdot111\ast1\\
=&\;\chi(3\smb1)\cdot\bigl(41+32+\chi(\mfq)\cdot311\bigr)
+\chi(21\smb1)\cdot\bigl(221+212+\chi(\mfq)\cdot2111\bigr)\\
&\;+\chi(12\smb1)\cdot\bigl(131+122+\chi(\mfq)\cdot1211\bigr)+\chi(111\smb1)\cdot\bigl(1121+1112\\
&\;+\chi(\mfq)\cdot11111\bigr)\\
=&\;\chi(3\smb1)\cdot(41+32)+\chi(3\smb\mfq\smb1)\cdot311+\chi(21\smb1)\cdot(221+212)\\
&\;+\chi(21\smb\mfq\smb1)\cdot2111+\chi(12\smb1)\cdot(131+122)+\chi(12\smb\mfq\smb1)\cdot1211\\
&\;+\chi(111\smb1)\cdot(1121+1112)+\chi(111\smb\mfq\smb1)\cdot11111.
\end{align*}
The P\"oppe product determines the quadratic and cubic terms in $\pi_5$. The quadratic terms 
contain the P\"oppe product of the natural numbers shown and all their compositions,
to which for the first term in the resulting product a `$1$' is added to the last letter
on the left factor, while for second term a `$1$' is added to the first letter of the second
factor, and for the third term, a letter $1$ is squeezed between the two factors.
For the cubic term, this process happens twice, for both P\"oppe products present,
and $\wdh{\und{1}}\ast\wdh{\und{1}}\ast\wdh{\und{1}}$ is the only monomial possible as other cubic combinations
will generate higher compositions that do not appear on the right in the expression for
$5\in\Rb\la\Cb\ra_\ast$ shown at the beginning of this example. In Table~\ref{table:KdV5}
we show the coefficients of all the compositions appearing in linear, quadratic and cubic
monomials of signature expansions that are in $\pi_5$.
Hence, referring to Table~\ref{table:KdV5}, the first column lists all the
compositions of $5$. The second column shows the signature coefficients
of the signature expansion corresponding to $\wdh{\und{5}}$. The third column
shows the signature coefficients of the signature expansion of the polynoimal $\wdh{\und{3}}\ast\wdh{\und{1}}$,
once all P\"oppe products have been expanded. The fourth and fifth colums show 
the corresponding signature coefficients in the respective
signature expansions of the polynomials $\wdh{\und{2}}\ast\wdh{\und{2}}$ and $\wdh{\und{1}}\ast\wdh{\und{3}}$.
The sixth column shows the signature coefficients of the signature expansion of
$\wdh{\und{1}}\ast\wdh{\und{1}}\ast\wdh{\und{1}}$. The seventh (and last) column shows the right-hand side
in the equation we are considering, namely $\pi_5=5$. So the question is,
can we find the vector of coefficients $C=(c_5,c_{31},c_{22},c_{13},c_{111})^{\mathrm{T}}$
such that $\pi_5=5$, i.e.\/ such that 
\begin{equation*}
c_5\cdot\wdh{\und{5}}+c_{31}\cdot\wdh{\und{3}}\ast\wdh{\und{1}}
+c_{22}\cdot\wdh{\und{2}}\ast\wdh{\und{2}}+c_{13}\cdot\wdh{\und{1}}\ast\wdh{\und{3}}
+c_{111}\cdot\wdh{\und{1}}\ast\wdh{\und{1}}\ast\wdh{\und{1}}=5\quad\Leftrightarrow\quad AC=B,
\end{equation*}
where $B$ is the $16$ component vector $B=(\chi(5),0,\ldots,0)^{\mathrm{T}}$?
The linear algebraic equation on the right results from equating the coefficients of all
$16$ compositions appearing in the equation on the left.
In the linear algebraic equation, the columns of the $16\times5$ matrix $A$ are
the second through sixth columns of the $\chi$-evaluated signatures shown in Table~\ref{table:KdV5}.
For example the second column of $A$ is
\begin{equation*}
  \bigl(\chi(3\smb1),\chi(3\smb1),0,0,\chi(3\smb\mfq\smb1),\chi(21\smb1),
  \chi(21\smb1),\chi(12\smb1),\ldots\bigr)^{\mathrm{T}}.
\end{equation*}
In the linear equation, columns one through five of $A$ are naturally
associated with the respective coefficients in rows one to five of $C$.

\begin{table}
\caption{Non-zero signature coefficients appearing
in the expansion of the \emph{P\"oppe polynomial} $\pi_5$; see Example~\ref{ex:fifthorder}.
The coefficients are the $\chi$-images of the signature entries shown.
Each column shows the factor contributions to the real coefficients of the compositions
of $\Cb(5)$ shown in the very left column, for each of the monomials in $\pi_5$ shown
across the top row. The final column represents the right-hand side of the equation $\pi_5=5$.}
\label{table:KdV5}
\begin{center}
\begin{tabular}{|l|ccccc|c|}
\hline
$\phantom{\biggl|}\Cb$&$\wdh{\und{5}}$&$\wdh{\und{3}}\ast\wdh{\und{1}}$
&$\wdh{\und{2}}\ast\wdh{\und{2}}$&$\wdh{\und{1}}\ast\wdh{\und{3}}$
&$\wdh{\und{1}}\ast\wdh{\und{1}}\ast\wdh{\und{1}}$&$B$\\
\hline
5 & 5 & &&&&5\\
\hline
41 & 41 & $3\smb1$ &         &         && 0\\
32 & 32 & $3\smb1$ & $2\smb2$ &         && 0\\
23 & 23 &         & $2\smb2$ & $1\smb3$ && 0\\
14 & 14 &         &         & $1\smb3$ && 0\\
\hline
311 & 311 & $3\smb\mfq\smb1$ & $2\smb11$        &                 &             & 0\\
221 & 221 & $21\smb1$        & $2\smb11$        & $1\smb21$        & $1\smb1\smb1$ & 0\\
212 & 212 & $21\smb1$        & $2\smb\mfq\smb2$ & $1\smb12$        & $1\smb1\smb1$ & 0\\
131 & 131 & $12\smb1$        &                 & $1\smb21$        & $1\smb1\smb1$ & 0\\
122 & 122 & $12\smb1$        & $11\smb2$        & $1\smb12$        & $1\smb1\smb1$ & 0\\
113 & 113 &                 & $11\smb2$        & $1\smb\mfq\smb3$ &             & 0\\
\hline
2111 & 2111 & $21\smb\mfq\smb1$ & $2\smb\mfq\smb11$ & $1\smb111$        & $1\smb1\smb\mfq\smb1$ & 0\\
1211 & 1211 & $12\smb\mfq\smb1$ & $11\smb11$        & $1\smb111$        & $1\smb1\smb\mfq\smb1$ & 0\\
1121 & 1121 & $111\smb1$        & $11\smb11$        & $1\smb\mfq\smb21$ & $1\smb\mfq\smb1\smb1$ & 0\\
1112 & 1112 & $111\smb1$        & $11\smb\mfq\smb2$ & $1\smb\mfq\smb12$ & $1\smb\mfq\smb1\smb1$ & 0\\
\hline
11111& 11111& $111\smb\mfq\smb1$& $11\smb\mfq\smb11$& $1\smb\mfq\smb111$& $1\smb\mfq\smb1\smb\mfq\smb1$& 0\\\hline
\end{tabular}
\end{center}
\end{table}

Let us now outline the stategy we use to solve the linear algebraic equation $AC=\bb$.
This linear system is overdetermined, there are $16$ equations and the components of $C$
represent the $5$ unknowns. Looking at Table~\ref{table:KdV5}, it makes sense to swap rows $6$
and $7$. Let us call the new coefficient matrix incorporating this swap $A'$. 
Note the vector $\bb$, whose only non-zero component is its first entry, is unaffected
by the row swap suggested.
For the moment, we ignore the first equation corresponding to the composition $5\in\Rb\la\Cb\ra_\ast$,
i.e.\/ we ignore the first row in the augmented matrix $[A'\,\bb]$.
We focus on the remaining system of $15$ homogeneous equations represented by the augmented matrix
$[A'(2,\ldots,16)\,O]$, where
$A'(2,\ldots,16)$ represents the submatrix of $A'$ only containing the rows $2$ through to $16$
and `$O$' is a column vector of $15$ zeros.
For any matrix, its row rank equals its column rank. Further, the rank of $A'(2,\ldots,16)$
and $[A'(2,\ldots,16)\,O]$ are the same since the last column of the latter is $O$.
Hence for the linear system of homogeneous equations represented by the augmented matrix
$[A'(2,\ldots,16)\,O]$, all but $5$ equations are redundant. Indeed let us identify,
a-priori, $5$ natural equations in $[A'(2,\ldots,16)\,O]$, namely those
represented by the first five rows of $[A'(2,\ldots,16)\,O]$; these correspond to
the rows in the original matrix $A$ identified by the compositions $41$, $32$, $23$, $14$ and $221$.
Since the leading diagonal entries in $A'(2,\ldots,16)$, i.e.\/ the diagonal entries in $A'(2,\ldots,6)$,
are all non-zero they can be used as pivots in a Gaussian elimination procedure to render the submatrix
$A'(7,\ldots,16)$ to the zero submatrix. In other words we can use Gaussian elimination on rows
$7$ through $16$ to ensure $A'(7,\ldots,16)=O_{10\times10}$.
The remaining homogeneous linear system represented by
the augmented matrix $[A'(2,\ldots,6)\,O]$ is given by 
\begin{equation*}
\begin{pmatrix}
  \chi(41)  & \chi(3\smb1)  & 0            & 0            & 0\\
  \chi(32)  & \chi(3\smb1)  & \chi(2\smb2)  & 0            & 0\\
  \chi(23)  & 0            & \chi(2\smb2)  & \chi(1\smb3)  & 0\\
  \chi(14)  & 0            & 0            & \chi(1\smb3)  & 0\\
  \chi(221) & \chi(21\smb1) & \chi(2\smb11) & \chi(1\smb21) & \chi(1\smb1\smb1)
\end{pmatrix}
\begin{pmatrix}c_5\\c_{31}\\c_{22}\\c_{31}\\c_{111}\end{pmatrix}
=\begin{pmatrix}0\\0\\0\\0\\0\end{pmatrix}.
\end{equation*}
Using Definition~\ref{def:signaturecharacter} and
that $\chi$ is a homomorphic map we have $\chi(41)=\chi(14)=5$,
$\chi(32)=\chi(23)=10$, $\chi(221)=30$, $\chi(3\smb1)=\chi(3)\chi(1)=1$,
$\chi(21\smb1)=\chi(21)\chi(1)=3$, and so forth.
Such a homogeneous linear system has a solution. The question is whether the
solution is the unique trivial solution $\cb=O$, or there is
a general solution with one or more free variables. See Meyer~\cite[p.~61]{Meyer}
for more details. We observe the final column in $A'(2,\ldots,6)$ just above
is a basic column that allows us to determine the final variable $c_{111}$
in terms of the other unknowns $c_5$, $c_{31}$, $c_{22}$ and $c_{13}$.
Hence we focus on the $4\times4$ subsystem for the latter four variables.
By Gaussian elimination the $4\times4$ homogeneous subsystem reduces to
$[A''\,O]$ where $A''$ is given by
\begin{equation*}
\begin{pmatrix}
  5 & 1 & 0  & 0  \\
  0 & 1 & -1 & 0  \\
  0 & 0 & 1  & -1 \\
  0 & 0 & 0  & 0 
\end{pmatrix},
\end{equation*}
and there is thus one free variable. The other three variables are
given explicitly, linearly, in terms of that free variable.
For example we can express the solution as $c_{31}=-5\,c_{5}$,
$c_{22}=-5\,c_{5}$ and $c_{13}=-5\,c_{5}$.
The free variable is fixed by the very first equation 
in the original linear system $A\cb=\bb$ which is $c_5=1$.
Since the final equation in the linear system $[A'(2,\ldots,6)\,O]$
giving $c_{111}$ is $c_{111}=-30\,c_5-3\,c_{31}-2\,c_{22}-3\,c_{13}$,
we find $c_{111}=10$. Hence we have solved the linear system of
equations $A'\cb=\bb$. The unique solution is $\cb=(1,-5,-5,-5,10)^{\mathrm{T}}$.
Hence we have shown that,
\begin{equation*}
\pi_5\coloneqq \wdh{\und{5}}-5\cdot(\wdh{\und{3}}\ast\wdh{\und{1}}
+\wdh{\und{2}}\ast\wdh{\und{2}}+\wdh{\und{1}}\ast\wdh{\und{3}})
+10\cdot(\wdh{\und{1}}\ast\wdh{\und{1}}\ast\wdh{\und{1}}).
\end{equation*}
\end{example}
We now turn to our main result.
Example~\ref{ex:fifthorder} illustrated the strategy we employ for the general order case.
\begin{definition}[P\"oppe polynomials]
For $n\in\mathbb N$, let
$\pi_n=\pi_n\bigl(\wdh{\und{1}},\wdh{\und{2}},\ldots,\wdh{\und{(n-2)}},\wdh{\und{n}}\bigr)$
denote a general polynomial consisting of a linear combination
of monomials of signature expansions of the following form:
\begin{equation*}
\pi_n\coloneqq\sum_{k=1}^{\frac12(n+1)}\sum_{a_1a_2\cdots a_k\in\Cb^\ast(n)}
c_{a_1a_2\cdots a_k}\cdot\wdh{\und{a_1}}\ast\wdh{\und{a_2}}\ast\cdots\ast\wdh{\und{a_k}},
\end{equation*}
where $\Cb^\ast(n)\subset\Cb(n)$ represents the subset
of compositions $w=a_1a_2\cdots a_k$ of $n$ such that
$a_1+a_2+\cdots+a_k=n-k+1$. The coefficients $c_{a_1a_2\cdots a_k}$ are real constants.
We call polynomials of this form in $\Rb\la\Cb\ra_\ast$, \emph{P\"oppe polynomials}.
\end{definition}
Note, in the definition of the P\"oppe polynomials,
$\Cb^\ast(n)$ is the correct subset of $\Cb(n)$ for the polynomial expansion $\pi_n$, as
we will equate $\pi_n$ to compositions of $n$ and each P\"oppe product `$\ast$'
adds `$1$' to the eventual compositions of $n$ produced in the
products $\wdh{\und{a_1}}\ast\wdh{\und{a_2}}\ast\cdots\ast\wdh{\und{a_k}}$.
The product $\wdh{\und{a_1}}\ast\wdh{\und{a_2}}\ast\cdots\ast\wdh{\und{a_k}}$ contains
$k-1$ such P\"oppe products. After expanding each of the P\"oppe polynomial factors
and expanding all the $k-1$ P\"oppe products between the resulting terms,
to guarantee all the compositions in the final sum are indeed compositions of $n$,
we must restrict the sum of the composition digits $a_1+a_2+\cdots a_k$ to be $n-(k-1)$.
Note the upper bound for $k$ is $\frac12(n+1)$ as this corresponds to the $k$-length
monomial $\und{1}\ast\und{1}\ast\cdots\ast\und{1}$ generating compositions of $n$.
\begin{theorem}[Main result: Integrability]\label{thm:mainresult}
For every odd natural number $n$, there exists a unique set of real coefficients
$\{c_w\colon w\in\Cb^\ast(n)\}$ such that in $\Rb\la\Cb\ra_\ast$, we have 
\begin{equation*}
n=\pi_n\bigl(\wdh{\und{1}},\wdh{\und{2}},\ldots,\wdh{\und{(n-2)}},\wdh{\und{n}}\bigr).
\end{equation*}
\end{theorem}
Before proving this theorem, we introduce a `descent ordering' on compositions as well as some useful notation.
\begin{definition}[Descent ordering of compositions]\label{def:naturalordering}
A composition $u\in\Cb$ precedes another composition $v\in\Cb$ and we write `$u\prec v$',
if the length of the compostion $u$, i.e.\/ the number of digits it contains,
is strictly less than the length of $v$. If $u$ and $v$ have the same length,
say $k$, so $u=u_1u_2\cdots u_k$ and $v=v_1v_2\cdots v_k$, then $u$
precedes $v$ if for some $\ell\in\{1,2,\dots,k\}$ we have $u_1=v_1$,
$u_2=v_2$, \ldots, $u_{\ell-1}=v_{\ell-1}$ and $u_\ell<v_\ell$. 
Otherwise we have $v\prec u$. The resulting ordering induced on $\Cb$,
is the \emph{descent ordering}.
\end{definition}
The descent ordering on $\Cb(n)$ is naturally transferred to subsets of $\Cb(n)$
such as $\Cb^\ast(n)$. We also use the notation $\Cb(n,k)$ to denote the set of $k$-part
compositions of $n$, i.e.\/ the subset of compositions $\Cb(n)$ whose length equals $k$:
$w\in\Cb(n,k)$ if $w\in\Cb(n)$ and $w=a_1a_2\cdots a_k$. We note, naturally
we have $\Cb(n)=\Cb(n,1)\cup\Cb(n,2)\cup\cdots\cup\Cb(n,n)$ and 
\begin{equation*}
\Cb^\ast(n)=\Cb^\ast(n,1)\cup\Cb^\ast(n-1,2)\cup\Cb^\ast(n-2,3)
\cup\cdots\cup\Cb^\ast\bigl(\tfrac12(n+1),\tfrac12(n+1)\bigr).
\end{equation*}
We also use the following notation: $\Cb(n,1)=n$; $\Cb(n,2)=\cup_{k=1}^{n-1}(n-k)\,\Cb(k,1)$ and then
\begin{equation*}
\Cb(n,3)=\bigcup_{k=1}^{n-2}(n-k-1)\,\Cb(k+1,2)\quad\text{and}\quad
\Cb(n,4)=\bigcup_{k=1}^{n-3}(n-k-2)\,\Cb(k+2,3),
\end{equation*}
and so forth. The notation $(n-k)\,\Cb(k,1)$ indicates the set of all compositions
of length $2$ which start with the digit `$(n-k)$'.
Similarly, we use the notation, $\wdh{\und{\Cb}}^\ast(n,1)=\wdh{\und{n}}$ and: 
\begin{align*}
\wdh{\und{\Cb}}^\ast(n-1,2)&=\bigcup_{k=1}^{n-2}\wdh{\und{(n-k-1)}}\ast\wdh{\und{\Cb}}^\ast(k,1), \\
\wdh{\und{\Cb}}^\ast(n-2,3)&=\bigcup_{k=1}^{n-4}\wdh{\und{(n-k-3)}}\ast\wdh{\und{\Cb}}^\ast(k+1,2), \\
\wdh{\und{\Cb}}^\ast(n-3,4)&=\bigcup_{k=1}^{n-6}\wdh{\und{(n-k-5)}}\ast\wdh{\und{\Cb}}^\ast(k+2,3), 
\end{align*}
and so forth. Here the notation $\wdh{\und{(n-k-3)}}\ast\wdh{\und{\Cb}}^\ast(k+1,2)$, for example,
indicates the the set of terms of the form $\wdh{\und{(n-k-3)}}\ast\wdh{\und{a}}_2\ast\wdh{\und{a}}_3$
where $a_2a_3\in\Cb^\ast(k+1,2)$.
Of course we can substitute any of these $\ell$-part composition
expressions into the $(\ell+1)$th one.
\begin{remark}
Note, with a slight abuse of notation, we have introduced the sets 
$\und{\Cb}^\ast(n)$ to denote compositions which we write in the
form $\und{a}_1\ast\und{a}_2\ast\cdots\ast\und{a}_k$, 
which have $k$-parts with $k$ ranging from $1$ to $\tfrac12(n+1)$,
where each factor $\und{a}_i$ is a signature polynomial corresponding to the
integer $a_i$, and $a_1a_2\cdots a_k\in\Cb^\ast(n,k)$.
In the proof of Theorem~\ref{thm:mainresult} just below,
we will refer to `columns' associated with the coefficients $c_{a_1\cdots a_k}$ in $\pi_n$,
with $a_1\cdots a_k\in\Cb^\ast(n,k)$. However each such coefficient
is allied with a monomial $\und{a}_1\ast\cdots\ast\und{a}_k$ in $\pi_n$.
In the proof, for convenience and brevity, we will equally refer to
the relevant `columns' either by the parametrising label $a_1\cdots a_k\in\Cb^\ast(n,k)$
or equivalently by the label $\und{a}_1\ast\cdots\ast\und{a}_k\in\und{\Cb}^\ast(n,k)$.
\end{remark}
We are now in a position to prove Theorem~\ref{thm:mainresult}.
\begin{proof}[of Theorem~\ref{thm:mainresult}]
We use Example~\ref{ex:fifthorder} and Table~\ref{table:KdV5}
as an example template for our strategy. Imagine we construct
the array in Table~\ref{table:KdV5} for the general integer case $n\in\Nb$
rather than the $n=5$ case shown in the table. We parametrise
the rows of the table using the set of compositions $\Cb(n)$, ordering
the compositions according to the descent ordering given in
Definition~\ref{def:naturalordering}. This naturally splits the
rows into blocks of compositions of the same length (i.e.\/ number of parts)
according to the decomposition,
\begin{equation*}
\Cb(n)=\Cb(n,1)\cup\Cb(n,2)\cup\cdots\cup\Cb(n,n).
\end{equation*}
This is the highest level coarse-grained block decomposition of $\Cb(n)$.
We parametrise the columns of the table using the set of decompositions $\Cb^\ast(n)$.
We again use descent ordering to parametrise the columns,
and this naturally splits the columns into blocks of compositions of $\Cb^\ast(n)$
with the same number of parts according to the decomposition,
\begin{equation*}
\Cb^\ast(n)=\Cb^\ast(n,1)\cup\Cb^\ast(n-1,2)\cup\Cb^\ast(n-2,3)
\cup\cdots\cup\Cb^\ast\bigl(\tfrac12(n+1),\tfrac12(n+1)\bigr).
\end{equation*}
Recall our goal is to prove there exists a unique set of real coefficients
$\{c_w\colon w\in\Cb^\ast(n)\}$ such that in $\Rb\la\Cb\ra_\ast$, we have 
\begin{equation*}
\sum_{k=1}^{\frac12(n+1)}\sum_{a_1a_2\cdots a_k\in\Cb^\ast(n,k)}
c_{a_1a_2\cdots a_k}\cdot\wdh{\und{a}}_1\ast\wdh{\und{a}}_2\ast\cdots\ast\wdh{\und{a}}_k=n.
\end{equation*}
Each column of the overall table corresponds to a term
$\wdh{\und{a}}_1\ast\wdh{\und{a}}_2\ast\cdots\ast\wdh{\und{a}}_k$
and its associated coefficient $c_{a_1a_2\cdots a_k}$ from the sum on the left.
The \emph{final column} in the overall table represents the right-hand side in
the equation above which is merely the single letter composition $n\in\Rb\la\Cb\ra_\ast$.
Each term $\wdh{\und{a}}_1\ast\wdh{\und{a}}_2\ast\cdots\ast\wdh{\und{a}}_k$, once the
P\"oppe products have been evaluated, generates a linear combination of
compositions from $\Cb(n)$. Within each column corresponding to
$a_1a_2\cdots a_k\in\Cb^\ast(n)$, each row contains the coefficient of
the corresponding composition generated in the evaluation of
the product $\wdh{\und{a}}_1\ast\wdh{\und{a}}_2\ast\cdots\ast\wdh{\und{a}}_k$. 
In other words, given a composition representing a given row in the table,
the element of that row in the column associated with a composition $a_1a_2\cdots a_k\in\Cb^\ast(n)$,
contains the $\chi$-image of the corresponding signature of the composition corresponding to that row.
Hence the resulting overall table, just like Table~\ref{table:KdV5}, represents
a linear equation of the form
\begin{equation*}
A\,\cb=\bb,
\end{equation*}
where $\cb$ is the vector of unknown coefficients $\{c_w\colon w\in\Cb^\ast(n)\}$,
in descent order so $\cb=(c_n,c_{(n-1)1},c_{(n-2)2},\ldots,c_{1(n-1)},c_{(n-2)11},\ldots)^{\mathrm{T}}$.
Since the number of $k$-part compositions of $n$ is $(n-1)$ choose $(k-1)$, the length
of the vector $\cb$ is 
\begin{equation*}
|\cb|\coloneqq\sum_{k=1}^{\frac12(n+1)}\begin{pmatrix}n-(k-1)\\k\end{pmatrix}.
\end{equation*}
The vector $\bb$ contains all zeros apart from the very first element
which is the $\chi$-image of the composition element $n\in\Cb(n)$,
i.e.\/ $\bb=\bigl(\chi(n),0,0,\ldots\bigr)^{\mathrm{T}}$.
Naturally we have $\chi(n)=1$.
The number of rows in $A$ and the length of $\bb$ is $2^{n-1}$,
the total number of compositions of $n$. Hence $A$ is a matrix of size $2^{n-1}\times|\cb|$,
with $|\cb|<2^{n-1}$ for $n\geqslant 2$.

As already indicated, we can decompose rows of $A$ into
blocks of compositions in $\Cb(n,k')$ for $k'=1,\ldots,n$, and we can
decompose columns of $A$ into blocks of compositions in
$\Cb^\ast(n-k+1,k)$ for $k=1,\ldots,\frac12(n+1)$.
We note the term $\wdh{\und{a}}_1\ast\wdh{\und{a}}_2\ast\cdots\ast\wdh{\und{a}}_{k}$
with $a_1a_2\cdots a_{k}\in\Cb^\ast(n-k+1,k)$, once the P\"oppe products have been evaluated,
only generates compositions in $\Cb(n)$ with $k$ or more parts.
Thus in terms of these blocks, parametrised by $(k',k)$, the matrix $A$ is lower block triangular.
Let us now proceed through the column blocks $\Cb^\ast(n-k+1,k)$ for $k=1$, $k=2$,
and so forth, to systematically determine the coefficients $c_{a_1a_2\cdots a_k}$
corresponding to each block.

First consider the case $k=1$, corresponding to the column block $\Cb^\ast(n,1)$
and thus the element $\wdh{\und{n}}$ and coefficient $c_n$. Since we know
$\wdh{\und{n}}=\sum_{w\in\Cb(n)}\chi(w)\cdot w$, this means the first row of this single
column block contains $\chi(n)=1$, while the subsequent rows corresponding to
$w\in\Cb(n)$ contain the coefficient $\chi(w)$, in descent order.
The first row of the linear system represented by the augmented matrix $[A\,\bb]$
represents the equation $c_n=1$. However as we did in Example~\ref{ex:fifthorder},
we temporarily ignore this knowledge. 

\begin{table}
\caption{Non-zero signature coefficients appearing in the $1$-part and $2$-part composition blocks of $A$,
ignoring the first row. The coefficients are the $\chi$-images of the signature entries shown.}
\label{table:secondblock}
\begin{center}
\begin{tabular}{|ccccccc|}
\hline
$\phantom{\biggl|}\wdh{\und{n}}$
&$\wdh{\und{(n-2)}}\ast\wdh{\und{1}}$
&$\wdh{\und{(n-3)}}\ast\wdh{\und{2}}$
&$\wdh{\und{(n-4)}}\ast\wdh{\und{3}}$
&\!$\cdots$\!
&$\wdh{\und{2}}\ast\wdh{\und{(n-3)}}$
&$\wdh{\und{1}}\ast\wdh{\und{(n-2)}}$\\
\hline
$(n-1)1$ & $(n-2)\smb1$ &              &              & \!$\cdots$\! &              &\\
$(n-2)2$ & $(n-2)\smb1$ & $(n-3)\smb2$ &              & \!$\cdots$\! &              &\\
$(n-3)3$ &              & $(n-3)\smb2$ & $(n-4)\smb3$ & \!$\cdots$\! &              &\\
$\vdots$ & $\vdots$     & $\vdots$     & $\vdots$     & \!$\ddots$\! & $\vdots$     & $\vdots$ \\
$2(n-2)$ &              &              &              & \!$\cdots$\! & $2\smb(n-3)$ & $1\smb(n-2)$ \\
$1(n-1)$ &              &              &              & \!$\cdots$\! &              & $1\smb(n-2)$ \\
\hline
\end{tabular}
\end{center}
\end{table}

Second we focus on the $2$-part composition sub-block of $A$ parametrised
by the columns $\Cb^\ast(n,1)\cup\Cb^\ast(n-1,2)$ and by the rows $\Cb(n,2)$, both in descent order. 
See Table~\ref{table:secondblock} where we show the non-zero elements in this block.
Note for the rows parametrised by $\Cb(n,2)$ all the columns in $A$ to the right
of the block contain zero entries and the corresponding entries in $\bb$ are all zero.
This block has size $(n-1)\times(n-1)$. If we focus on these rows in the
augmented system $[A\,\bb]$ then we have $(n-1)$ equations in the $(n-1)$ unknowns
$c_n$, $c_{(n-2)1}$, $c_{(n-3)2}$, \ldots, $c_{1(n-2)}$. We also observe that all the leading
diagonal entries in the block shown in Table~\ref{table:secondblock}, from
the top left $(1,1)$ entry to the lower right $(n-1,n-1)$ entry,
are all non-zero. Thus all these leading diagonal entries can be used as pivots
in a Gaussian elimination process to render all the entries in the corresponding
columns in the rows below this block to be zero. 
We can also solve the homogeneous linear system above,
represented by this block, to find the coefficients.
For $\ell=2,\ldots,(n-2)$ the homogeneous linear system of equations takes the form,
\begin{align*}
\chi\bigl((n-1)1\bigr)\,c_n+\chi\bigl((n-2)\smb1\bigr)\,c_{(n-2)1}&=0,\\
\chi\bigl((n-\ell)\ell\bigr)\,c_n+\chi\bigl((n-\ell)\smb(\ell-1)\bigr)\,c_{(n-\ell)(\ell-1)}
+\chi\bigl((n-\ell-1)\smb\ell\bigr)\,c_{(n-\ell-1)\ell}&=0,\\
\chi\bigl(1(n-1)\bigr)\,c_n+\chi\bigl(1\smb(n-2)\bigr)\,c_{1(n-2)}&=0.
\end{align*}
Note, apart from the coefficients of $c_n$ shown, all the other coefficients
of $c_{(n-2)1}$, $c_{(n-3)2}$, \ldots, $c_{1(n-2)}$ are equal to one.
Starting with the first equation above, we can iteratively solve for all the coefficients
in terms of $c_n$ giving, for $\ell=1,\ldots,(n-2)$, the relations
$\hat\upchi_{\ell}\,c_n+c_{(n-\ell-1)\ell}=0$, where
\begin{equation*}
\hat\upchi_\ell\coloneqq\chi\bigl((n-\ell)\ell\bigr)-\chi\bigl((n-\ell+1)(\ell-1)\bigr)
+\chi\bigl((n-\ell+2)(\ell-2)\bigr)-\cdots+(-1)^{\ell+1}\chi\bigl((n-1)1\bigr),
\end{equation*}
along with the final equation which is $\chi\bigl(1(n-1)\bigr)\,c_n+c_{1(n-2)}=0$.
We now have two linear homogeneous equations relating the unknowns $c_{1(n-2)}$ and $c_n$,
the equation for which $\ell=n-2$ and the final equation just mentioned.
These two equations are consistent if $\hat\upchi_{n-2}=\chi\bigl(1(n-1)\bigr)$. 
Note from its definition, $\chi\bigl(nk\bigr)$ is the binomial coefficient $n+k$ choose $n$,
and thus equal to $\chi\bigl(kn\bigr)$. Depending on whether $\ell$ is odd or even we have,
\begin{equation*}
\hat\upchi_\ell=\begin{cases} 2m+1, &~\text{if}~\ell=2m+1,\\ 2m-2, &~\text{if}~\ell=2m.\end{cases}
\end{equation*}
To see this, we observe that in the odd case we can utilise the symmetry $\chi(nk)=\chi(kn)$
and the differing signs in front of each as they appear in $\hat\upchi_\ell$, to cancel
all the terms except the final one given by $(-1)^{2m}\chi\bigl((2m)1\bigr)=2m+1$.
In the even case we observe
\begin{align*}
  \hat\upchi_{2m}&=\begin{pmatrix} 2m \\ 2 \end{pmatrix}-\begin{pmatrix} 2m \\ 3 \end{pmatrix}
  +\begin{pmatrix} 2m \\ 4 \end{pmatrix}-\cdots-\begin{pmatrix} 2m \\ 2m-1 \end{pmatrix}\\
  &=\Bigl((1-x)^{2m}-1+(2m)x-x^{2m}\Bigr)\Big|_{x=1}\\
  &=2m-2.
\end{align*}
Hence we observe when $n-2=2m+1$ we have $\hat\upchi_{n-2}=\hat\upchi_{2m+1}=2m+1=\chi\bigl(1(n-1)\bigr)$ 
and the homogeneous linear system of equations above is consistent with a single free
variable, namely $c_n$. On the other hand when $n-2=2m$ then $\hat\upchi_{2m}=2m-2$ and
this does not equal $\chi\bigl(1(n-1)\bigr)=2m$. Thus in the latter case there is no
solution to the linear homogeneous system $[A\,\bb]$ corresponding to the $2$-part composition rows.
Given the consistency in the odd order case, we now recall the equation from the very first
row of the linear system $[A\,\bb]$ which states $c_n=1$. When we include this information
we observe we can solve the top part of the linear system of equations $[A\,\bb]$, corresponding
to the single and then $2$-part composition rows, uniquely for the coefficients
$c_n$, $c_{(n-2)1}$, $c_{(n-3)2}$, \ldots, $c_{1(n-2)}$.

Third we now focus on solving the linear system $[A\,\bb]$ for all the coefficients
$c_w$ with $w\in\Cb^\ast(n)$ systematically. We assume $n$ is odd,
and in some sense we start again from the beginning and \emph{ignore} the information
given in the equation represented by the very first row, i.e.\/ that $c_n=1$.
We focus on the homogeneous linear system given by the remaining rows and represented by
$\bigl[A(2,\ldots,2^{n-1})\,O\bigr]$, where as before $A(2,\ldots,2^{n-1})$
represents the submatrix of $A$ consisting of rows $2$ through $2^{n-1}$. 
As is evident from the previous paragraph, the homogeneous linear system
of $(n-1)$ equations represented by $\bigl[A(2,\ldots,n)\,O\bigr]$
in the $(n-1)$ unknowns $c_n$, $c_{(n-2)1}$, $c_{(n-3)2}$, \ldots, $c_{1(n-2)}$
can be rendered into an \emph{upper} triangular form with non-zero entries
on the leading diagonal. Further, in the rows $2$ through $n$ considered,
the entries in the columns parametrised by
$\Cb^\ast(n-2,3)\cup\Cb^\ast(n-3,4)\cup\cdots\cup\Cb^\ast\bigl(\frac12(n+1),\frac12(n+1)\bigr)$
are all zero.

Continuing, we now focus on the $3$-part composition sub-block of $A$ parametrised
by the columns $\Cb^\ast(n-2,3)$ and by the rows $\Cb(n,3)$, both in descent order.
The reader may find Table~\ref{table:KdVn-3parts} in Appendix~\ref{sec:tables}
helpful in visualising the block structure we now discuss.
Note all the entries in $A$ above and to the right of this sub-block are zero.
Our goal here is to show, the $3$-part composition rectangular sub-block of $A$
is itself \emph{lower} triangular,
in the sense that by row swaps we rearrange the sub-block so that
each column has a non-zero pivot entry on the leading diagonal with all
the entries in the rows above the pivot equal to zero. If this is possible,
then we can use elementary row operations to render the entries in
\emph{all} the rows below the leading diagonal rows,
corresponding to any compositions with $3$- or more parts, to be zero.
Note we can use elementary row operations to render \emph{any} such redundant rows below
the $2$-part composition rows in the columns parametrised by the $1$- and $2$-part
composition columns, i.e.\/ by $\Cb^\ast(n,1)\cup\Cb^\ast(n-1,2)$, to have zero entries.
This is because the leading diagonal in the homogeneous linear system corresponding
to these columns has non-zero entries.
Thus using the leading diagonal rows in the $3$-part composition sub-block of $A$,
we can uniquely determine the unknowns $c_{w}$ with $w\in\Cb^\ast(n-2,3)$ in terms of $c_n$,
just as we achieved (see above) for the unknowns $c_{w}$ with $w\in\Cb^\ast(n-1,2)$.
Recall, retaining descent order, we express
\begin{equation*}
\Cb(n,3)\!=\!\!\bigcup_{k'=1}^{n-2}(n-k'-1)\,\Cb(k'+1,2)
\quad\text{and}\quad
\wdh{\und{\Cb}}^\ast(n-2,3)\!=\!\!\bigcup_{k=1}^{n-4}\wdh{\und{(n-k-3)}}\ast\wdh{\und{\Cb}}^\ast(k+1,2).
\end{equation*}
These two decompositions of $\Cb(n,3)$ and $\und{\Cb}^\ast(n-2,3)$ show we
can decompose the corresponding sub-block of $A$ into further, level~$2$ sub-blocks.
The level~$2$ sub-blocks are parametrised by $k'$ and $k$ with the ranges indicated above.
Further, each level~$2$ sub-block, for each $(k',k)$ pair, is, within themselves parametrised
by the $2$-part composition rows and columns $\Cb(k'+1,2)$ and $\und{\Cb}^\ast(k+1,2)$, respectively.
We proceed systematically, pairing up level~$2$ sub-blocks.
To begin, we observe the top left level~$2$ sub-block which corresponds
to the values $(k',k)=(1,1)$ and is a $1\times1$ level~$2$ sub-block
corresponding to the row $(n-2)11$ and column $(n-4)11$ must have a zero
entry as $\wdh{\und{(n-4)}}\ast\wdh{\und{1}}\ast\wdh{\und{1}}$ cannot generate the composition $(n-2)11$.
Indeed, by the same argument, all the entries in the remaining columns in this row,
i.e.\/ to the right in terms of descent order, must also be zero.
Hence the homogeneous linear equation represented by this row only contains
the unknowns $c_{w}$ with $w\in\und{\Cb}^\ast(n,1)\cup\und{\Cb}^\ast(n-1,2)$.
Since we already have a set of $(n-1)$ homogeneous linear equations in these
$(n-1)$ unknowns with non-zero diagonal entries, the new equation must be
a linear combination of these equations, and thus contains no new information.
Alternatively, as discussed above, all the entries in this row in preceding
columns corresponding to the $1$- and $2$-part compositions $\und{\Cb}^\ast(n,1)\cup\und{\Cb}^\ast(n-1,2)$
can be rendered zero by elementary row operations.
Thus again the whole row contains no new information and we can ignore it, which we do henceforth.
Further since for $k=1,\ldots,n-4$ the elements in $\wdh{\und{(n-k-3)}}\ast\wdh{\und{\Cb}}^\ast(k+1,2)$
can only generate compositions $(n-k'-1)\,\Cb(k'+1,2)$ when $(n-k'-1)-(n-k-3)\leqslant1$,
which is equivalent to $k\leqslant k'-1$, we deduce that for $k\geqslant k'$, all the
level~$2$ sub-blocks are zero. Hence in terms of level~$2$ sub-blocks the $3$-part composition
rectangular sub-block of $A$ is lower triangular. Now we just need to demonstrate each
level~$2$ sub-block parametrised by $(k',k)=(k+1,k)$ itself is lower triangular.
Each such level~$2$ sub-block consists of the
rows $(n-k-2)\,\Cb(k+2,2)$ and columns $\wdh{\und{(n-k-3)}}\ast\wdh{\und{\Cb}}^\ast(k+1,2)$.
The first letters in these compositions are fixed and the rows and columns
within the level~$2$ sub-block are parametrised by $\Cb(k+2,2)$ and $\wdh{\und{\Cb}}^\ast(k+1,2)$.
It is straightforward to check any such level~$2$ rectangular sub-blocks are indeed
lower triangular with non-zero entires on the diagonal as follows.
Recalling the structure of the level~$1$ sub-block
of $2$-part composition rows and columns from above,
we know that each such level~$2$ rectangular sub-block has size $(k+1)\times k$.
The final row has a non-zero entry in the final column of the sub-block.
However as we proceed through $k=1,2,3,\ldots$ we observe that each 
homogeneous linear equation represented by those final rows must be
a linear combination of the homogneous linear equations that precede it,
as those preceding equations are a square system of linear combinations of the same 
same set of unknowns. Hence each such final row is redundant.
Hence the $3$-part composition rectangular sub-block of $A$
is lower triangular and we can solve for the unknowns $c_{w}$ with $w\in\Cb^\ast(n-2,3)$
linearly in terms of $c_n$.
Note that all entries in the columns $\und{\Cb}^\ast(n-2,3)$ below the rows $\Cb(n,3)$
can be rendered zero by elementary row operations.

We now briefly outline the procedure for determining 
the unknowns $c_{w}$ with $w\in\und{\Cb}^\ast(n-3,4)$ linearly in terms of $c_n$.
As we do so, the nature of the procedure for determining the
remaining unknowns becomes apparent, and reveals itself to be recursive and straightforward. 
We focus on the $4$-part composition sub-block of $A$ parametrised
by the columns $\und{\Cb}^\ast(n-3,4)$ and by the rows $\Cb(n,4)$, both in descent order.
The reader may find Table~\ref{table:KdVn-4parts} in Appendix~\ref{sec:tables}
helpful in visualising the block structure we discuss herein.
All the entries in $A$ above and to the right of this sub-block are zero, since
elements in $\und{\Cb}^\ast(n-3,4)$ can only generate compositions in $\Cb(n)$ with
four or more parts. Our goal is to show, the $4$-part composition rectangular
sub-block of $A$ is itself \emph{lower} triangular, as we did for the 
$3$-part composition rectangular sub-block of $A$. Again as above,
if this is possible, then we can use elementary row operations to render the entries in
\emph{all} the rows below the leading diagonal rows,
corresponding to any compositions with $4$- or more parts, to be zero.
The leading diagonal rows in the $4$-part composition sub-block of $A$,
uniquely determine the unknowns $c_{w}$ with $w\in\Cb^\ast(n-3,4)$ in terms of $c_n$.
Recall, as always retaining descent order, we have
\begin{equation*}
\Cb(n,4)\!=\!\!\bigcup_{k'=1}^{n-3}(n-k'-2)\,\Cb(k'+2,3)
\quad\text{and}\quad
\wdh{\und{\Cb}}^\ast(n-3,4)\!=\!\!\bigcup_{k=1}^{n-6}\wdh{\und{(n-k-5)}}\ast\wdh{\und{\Cb}}^\ast(k+2,3).
\end{equation*}
Thus we observe we decompose the sub-block of $A$ parametrised by 
the rows $\Cb(n,4)$ and columns $\Cb^\ast(n-3,4)$ into level~$2$ sub-blocks.
The level~$2$ sub-blocks are parametrised by $k'$ and $k$ as indicated above.
Each level~$2$ sub-block, for each $(k',k)$ pair, is, within themselves parametrised
by the $3$-part composition rows and columns $\Cb(k'+2,3)$ and $\Cb^\ast(k+2,3)$, respectively.
And each of these level~$2$ sub-blocks can be decomposed into further level~$3$ sub-blocks
of $2$-part compositions, as we demonstrated in the previous paragraph. 
We start by showing that in terms of level~$2$ sub-blocks, the $4$-part composition rectangular
sub-block of $A$ is lower triangular. 
Observe the top left level~$2$ sub-block which corresponds
to the values $(k',k)=(1,1)$ is a $1\times1$ level~$2$ sub-block
corresponding to the row $(n-3)111$ and column $(n-6)111$ must have a zero
entry as $\wdh{\und{(n-6)}}\ast\wdh{\und{1}}\ast\wdh{\und{1}}\ast\wdh{\und{1}}$
cannot generate the composition $(n-3)111$.
Indeed, by the same argument, all the entries in the remaining columns in this row,
i.e.\/ to the right in terms of descent order, must also be zero.
Hence the homogeneous linear equation represented by this row does not generate
any new information. Next, the sub-block which corresponds
to the values $(k',k)=(2,1)$ is a $3\times1$ level~$2$ sub-block
corresponding to the rows $(n-4)211$, $(n-4)121$ and $(n-4)112$ and column $(n-6)111$ must have zero
entries as $\wdh{\und{(n-6)}}\ast\wdh{\und{1}}\ast\wdh{\und{1}}\ast\wdh{\und{1}}$ cannot generate any of
the compositions $(n-4)211$, $(n-4)121$ and $(n-4)112$. By similar arguments these rows
do not generate any new information either. Indeed, we observe that for 
$k=1,\ldots,n-6$ the elements in $\wdh{\und{(n-k-5)}}\ast\wdh{\und{\Cb^\ast}(k+2,3)}$
can only generate compositions $(n-k'-2)\,\Cb(k'+2,3)$ when $(n-k'-2)-(n-k-5)\leqslant1$,
which is equivalent to $k\leqslant k'-2$, we deduce that for $k\geqslant k'-1$, all the
level~$2$ sub-blocks are zero. Thus, in terms of level~$2$ sub-blocks, the $4$-part composition
rectangular sub-block of $A$ is lower triangular---either by deleting the rows corresponding
to the first two row sub-blocks or by row-swapping them to rows sufficiently far down the matrix.
Now we demonstrate that each level~$2$ sub-block with $k'=k+2$ is itself lower triangular.
Indeed we have already demonstrated this in the last paragraph, which recursively used that
each of these level~$2$, $3$-part composition sub-blocks can be further decomposed
into level~$3$, $2$-part composition sub-blocks, which are lower triangular.
Consequently we can solve for the unknowns $c_{w}$ with $w\in\Cb^\ast(n-3,4)$
linearly in terms of $c_n$.

Finally we discuss the general case, for any $5\leqslant\ell\leqslant\frac12(n+1)$,
for the $\ell$-part composition sub-block of $A$ parametrised
by the columns $\Cb^\ast(n-\ell+1,\ell)$ and by the rows $\Cb(n,\ell)$, both in descent order.
All the entries in $A$ above and to the right of this sub-block are zero, since
elements in $\Cb^\ast(n-\ell,\ell)$ can only generate compositions in $\Cb(n)$ with
$\ell$-parts or more. Assume that sub-blocks parametrised by rows $\Cb(n,\ell-1)$
and columns $\Cb^\ast(n-\ell+2,\ell-1)$, i.e.\/ at the previous stage `$\ell-1$',
are lower triangular. Now, again retaining descent order, we have
\begin{align*}
\Cb(n,\ell)&=\bigcup_{k'=1}^{n-\ell+1}(n-k'-\ell+2)\,\Cb(k'+\ell-2,\ell-1)
\intertext{and}
\wdh{\und{\Cb}}^\ast(n-\ell+1,\ell)
&=\bigcup_{k=1}^{n-2\ell+2}\wdh{\und{(n-2\ell+3-k)}}\ast\wdh{\und{\Cb}}^\ast(k+\ell-2,\ell-1).
\end{align*}
Thus 
we decompose the sub-block of $A$ parametrised by 
the rows $\Cb(n,\ell)$ and columns $\Cb^\ast(n-\ell+1,\ell)$ into level~$2$ sub-blocks.
The level~$2$ sub-blocks are parametrised by $k'$ and $k$ as indicated above.
For $k=1,\ldots,n-2\ell+2$ the elements in $\wdh{\und{(n-2\ell+3-k)}}\ast\wdh{\und{\Cb}}^\ast(k+\ell-2,\ell-1)$
can only generate compositions $(n-k'-\ell+2)\,\Cb(k'+\ell-2,3)$ when
$(n-k'-\ell+2)-(n-2\ell+3-k)\leqslant1$, which is equivalent to $k\leqslant k'-\ell+2$.
Hence for $k\geqslant k'-\ell+3$, all the
level~$2$ sub-blocks are zero. Thus, in terms of level~$2$ sub-blocks,
the $\ell$-part composition rectangular sub-block of $A$ is lower triangular---or
can be transformed as such by elementary row swaps. 
Now using that we assumed the sub-blocks parametrised by rows $\Cb(n,\ell-1)$
and columns $\Cb^\ast(n-\ell+2,\ell-1)$ are lower triangular, the
sub-block of $A$ parametrised
by the columns $\Cb^\ast(n-\ell+1,\ell)$ and the rows $\Cb(n,\ell)$ must
be lower triangular. Hence we can solve for the unknowns $c_{w}$ with $w\in\Cb^\ast(n-\ell+1,\ell)$
linearly in terms of $c_n$. By induction,
this applies for all values of $\ell$ up to and including $\ell=\frac12(n+1)$.
Finally now utilising that $c_n=1$ generates a unique solution set for $c_w$ for all $w\in\Cb^\ast(n)$.\qed
\end{proof}
\begin{remark}[Even order]\label{rmk:evenordercase}
For every even natural number $n$, no such set of real coefficients exist.
This is demonstrated in the proof above in the component addressing the $2$-part composition
sub-blocks of $A$.
\end{remark}
\begin{remark}[Linear combinations of fields]
Since we can expand each of the single letter compositions in $\mu_3\cdot3+\mu_5\cdot5+\mu_7\cdot7+\cdots$,
in terms of $\mu_3\cdot\pi_3+\mu_5\cdot\pi_5+\mu_7\cdot\pi_7+\cdots$, we deduce that we can
solve/linearise, in the manner we have outlined, any nonlinear
evolutionary partial differential equation generated by a linear combination of fields from the
non-commutative primitive Korteweg--de Vries hierarchy.
\end{remark}

\section{Lax Hierarchy}\label{sec:Laxhierarchy}
Our goal in this section is to show how the non-commutative Lax hierarchy
is easily and naturally deduced using the signature expansions 
in the P\"oppe algebra $\Rb\la\Cb\ra_\ast$ from
Sections~\ref{sec:kernelalgebras} and \ref{sec:existenceanduniqueness}.
We derive the well-known non-commutative Lax hierarchy,
see for example Carillo and Schoenlieb~\cite{CSII},
and show that at each odd order, the underlying iteration generates 
nonlinear terms that can be expressed as a linear combination 
of monomials of the form $\wdh{\und{a_1}}\ast\wdh{\und{a_2}}\ast\cdots\ast\wdh{\und{a_k}}$
with $a_1a_2\ldots a_k\in\Cb^\ast(n)$. By our main result Theorem~\ref{thm:mainresult}
from the last section on the existence and uniqueness of the Korteweg--de Vries
hierarchy, we deduce the non-commutative hierarchy we developed therein,
and the non-commutative Lax hierarchy, are one and the same. 

We proceed by deriving the first part of the result on the non-commutative Lax hierarchy,
before stating the main result of this section, and then subsequently proving the second part. 
For a given odd natural number $n\in\mathbb N$, we assume there exists a polynomial
$\hat\pp_n=\hat\pp_n\bigl(\wdh{\und{1}},\wdh{\und{2}},\ldots,\wdh{\und{(n-2)}}\bigr)$
such that the composition element $n\in\Rb\la\Cb\ra_\ast$ can be expressed in the form
\begin{equation*}
n=\wdh{\und{n}}+\hat\pp_n\bigl(\wdh{\und{1}},\wdh{\und{2}},\ldots,\wdh{\und{(n-2)}}\bigr).
\end{equation*}
For $n=3$, such a polynomial $\hat\pp_3=\hat\pp_3\bigl(\und{1}\bigr)$ exists as we know
$3=\wdh{\und{3}}-3\cdot\wdh{\und{1}}\ast\wdh{\und{1}}$, representing the potential Korteweg--de Vries equation,
i.e.\/ $\hat\pp_3=-3\cdot\wdh{\und{1}}\ast\wdh{\und{1}}$. We now apply the derivation operator $\md$
to the expression for the composition element $n\in\Rb\la\Cb\ra_\ast$ just above, twice.
Recall the action of $\md$ on monomials in $\Rb\la\Cb\ra_\ast$ from Definition~\ref{def:derivation}.
Starting with $\md(n)=1n+(n+1)+n1$ and then $\md(1n)=2\cdot(11n)+2n+1(n+1)+1n1$
and $\md(n1)=1n1+(n+1)1+2\cdot(n11)+n2$, the successive application of $\md$ generates,
first that $(n+1)+(1n+n1)=\wdh{\und{(n+1)}}+\md\hat\pp_n$, and then:
\begin{equation*}
(n+2)+2\cdot\bigl(11n+1(n+1)+1n1+(n+1)1+n11\bigr)+(2n+n2)=\wdh{\und{(n+2)}}+\md^2\hat\pp_n.
\end{equation*}
Next, in the last expression, looking at the terms on the left, not including the first term, 
and recalling the properties of the P\"oppe product, we observe
\begin{equation*}
1\ast n+n\ast 1=2n+1(n+1)+2\cdot(11n)+(n+1)1+n2+2\cdot(n11).
\end{equation*}
Substituting this result into the previous expression and rearranging, we find,
\begin{equation*}
(n+2)=\wdh{\und{(n+2)}}+\md^2\hat\pp_n-(1\ast n+n\ast 1)-\bigl(1(n+1)+(n+1)1+2\cdot(1n1)\bigr).
\end{equation*}
To find a closed form expression for the terms `$1(n+1)+(n+1)1+2\cdot(1n1)$' in terms
of a P\"oppe polynomial, we use the following result relating the P\"oppe product
and the derivation operation $\md$. 
\begin{lemma}\label{lemma:Poppeproductandderivation}
Given two arbitrary compositions $ua,bv\in\Rb\la\Cb\ra_\ast$ where we 
distinguish the final and beginning single letters $a$ and $b$ as shown,  
we have:
\begin{equation*}
ua\ast bv=\md(uabv)-(\md u)abv-uab(\md v)+ua1bv.
\end{equation*}
\end{lemma}
\begin{proof}
From the definition of $\md$ we have
$\md(uabv)=(\md u)abv+u\bigl(\md(ab)\bigr)v+uab(\md v)$. Thus from the definition
of the P\"oppe product, the result follows as we observe
$ua\ast bv$ equals $u(a+1)bv+ua(b+1)v+2\cdot(ua1bv)$,
which in turn can be expressed as $u\bigl(\md(ab)\bigr)v+ua1bv$,
which itself in turn can be expressed as $\md(uabv)-(\md u)abv-uab(\md v)+ua1bv$.\qed 
\end{proof}
Now we observe that since $\md(n)=1n+(n+1)+n1$ then utilising
Lemma~\ref{lemma:Poppeproductandderivation} with $ua=1$ and
then $bv$ successively equal to $1n$, $(n+1)$ and $n1$, and vice-versa,
we have:
\begin{align*}
1\ast\md(n)+\md(n)\ast1
=&\;\md\bigl(1(n+1)+(n+1)1+2\cdot(1n1)+11n+n11\bigr)\\
&\;-\Bigl(11\md(n)+\md(n)11+1n\md(1)+\md(1)n1+2\cdot\bigl(1(n+1)1\bigr)\Bigr)\\
=&\;\md\bigl(1(n+1)+(n+1)1+2\cdot(1n1)\bigr)\\
&\;-\Bigl(21n+12n+2\cdot(111n)-2n1-1(n+1)1-2\cdot(11n1)\\
&\;+n21+n12+2\cdot(n111)-1(n+1)1-1n2-2\cdot(1n11)\Bigr)\\
=&\;\md\bigl(1(n+1)+(n+1)1+2\cdot(1n1)\bigr)\\
&\;-1\ast(1n)+1\ast(n1)-(n1)\ast1+(1n)\ast1\\
=&\;\md\bigl(1(n+1)+(n+1)1+2\cdot(1n1)\bigr)
-[1,1n]_\ast+[1,n1]_\ast, 
\end{align*}
where in the second step we expanded the derivation terms $\md(11n+n11)$,
and for $u,v\in\Rb\la\Cb\ra_\ast$ the expression $[u,v]_\ast$ denotes
the P\"oppe product commutator of $u$ and $v$, i.e.\/ $[u,v]_\ast\coloneqq u\ast v-v\ast u$.
Using Lemma~\ref{lemma:Poppeproductandderivation} with $u$ and $v$ as the empty
compositions and $a=1$ and $b=n$, as well as vice-versa, we deduce that,
\begin{equation*}
1\ast n-n\ast 1=\md(1n-n1)\qquad\Leftrightarrow\qquad
1n-n1=\md^{-1}\bigl(1\ast n-n\ast1\bigr).
\end{equation*}
Subsituting this into the expression involving $-[1,1n]_\ast+[1,n1]_\ast\equiv-[1,1n-n1]_\ast$
just above and rearranging, we observe,
\begin{equation*}
1(n+1)+(n+1)1+2\cdot(1n1)
=\md^{-1}\Bigl(\bigl(1\ast\md(n)+\md(n)\ast1\bigr)+\bigl[1,\md^{-1}[1,n]_\ast\bigr]_\ast\Bigr).
\end{equation*}
By substituting this result into the expression for $(n+2)\in\Rb\la\Cb\ra_\ast$
just preceding Lemma~\ref{lemma:Poppeproductandderivation} 
we have established the first part of the following result.
The second part is established in the proof of the theorem that follows.
\begin{theorem}[Non-Commutative Lax hierarchy]\label{thm:Laxhierarchy}
Let $n\in\mathbb N$ be a given odd natural number with $n\geqslant 3$.
Assume there exists a polynomial
$\hat\pp_n=\hat\pp_n\bigl(\wdh{\und{1}},\wdh{\und{2}},\ldots,\wdh{\und{(n-2)}}\bigr)$
such that for the composition element $n\in\Rb\la\Cb\ra_\ast$ we have
\begin{equation*}
n=\wdh{\und{n}}+\hat\pp_n\bigl(\wdh{\und{1}},\wdh{\und{2}},\ldots,\wdh{\und{(n-2)}}\bigr).
\end{equation*}
For $n=3$, such a polynomial $\hat\pp_3=-3\cdot\wdh{\und{1}}\ast\wdh{\und{1}}$ exists.
By successively applying the derivation operation twice we deduce
\begin{equation*}
(n+2)=\wdh{\und{(n+2)}}+\md^2\hat\pp_n-(1\ast n+n\ast 1)-\md^{-1}\bigl(1\ast\md(n)+\md(n)\ast1\bigr)
-\bigl[1,\md^{-1}[1,n]_\ast\bigr]_\ast.
\end{equation*}
In addition there exists a polynomial
$\tp\pp_n=\tp\pp_n\bigl(\wdh{\und{1}},\wdh{\und{2}},\ldots,\wdh{\und{n}}\bigr)$
in $\Rb\la\Cb\ra_\ast$ such that 
\begin{equation*}
\md^{-1}\bigl(1\ast\md(n)+\md(n)\ast1\bigr)+\bigl[1,\md^{-1}[1,n]_\ast\bigr]_\ast
=\tp\pp_n\bigl(\wdh{\und{1}},\wdh{\und{2}},\ldots,\wdh{\und{n}}\bigr).
\end{equation*}
\end{theorem}
\begin{remark}
Note we can substitute for `$n$' in the term $(1\ast n+n\ast 1)$ in the
second statement in Theorem~\ref{thm:Laxhierarchy} from the first statement 
in the theorem.
\end{remark}
\begin{remark}
The expression for $(n+2)$ in Theorem~\ref{thm:Laxhierarchy} matches
that for the non-commutative potential Korteweg--de Vries hierarchy
in Carillo and Schoenlieb~\cite[Sec.~II]{CSII}.
\end{remark}
\begin{proof}
We establish the existence of the polynomial
$\tp\pp_n=\tp\pp_n\bigl(\wdh{\und{1}},\wdh{\und{2}},\ldots,\wdh{\und{n}}\bigr)$
for the terms shown in the theorem. From the arguments preceding
the theorem, we know that
$\md^{-1}\bigl(1\ast\md(n)+\md(n)\ast1\bigr)+\bigl[1,\md^{-1}[1,n]_\ast\bigr]_\ast\equiv1(n+1)+(n+1)1+2\cdot(1n1)$.
Hence, equivalently, we can show there exists a polynomial
$\tp\pp_n=\tp\pp_n\bigl(\wdh{\und{1}},\wdh{\und{2}},\ldots,\wdh{\und{n}}\bigr)$ such that
\begin{equation*}
1(n+1)+(n+1)1+2\cdot(1n1)
=\tp\pp_n\bigl(\wdh{\und{1}},\wdh{\und{2}},\ldots,\wdh{\und{n}}\bigr).
\end{equation*}
The argument required to establish this is analogous to that
used in the proof of Theorem~\ref{thm:mainresult} in
Section~\ref{sec:existenceanduniqueness}, with some slight modifications,
particularly in the first steps. Suppose 
$\tp\pp_n=\tp\pp_n\bigl(\wdh{\und{1}},\wdh{\und{2}},\ldots,\wdh{\und{n}}\bigr)$ has the form
\begin{equation*}
\tp\pp_n\coloneqq\sum_{k=1}^{\frac12(n+3)}\sum_{a_1a_2\cdots a_k\in\Cb^\ast(n+2)}
\tp c_{a_1a_2\cdots a_k}\cdot\wdh{\und{a}}_1\ast\wdh{\und{a}}_2\ast\cdots\ast\wdh{\und{a}}_k.
\end{equation*}
Note, implicitly, we exclude the possibility of the single term monomial `$\wdh{\und{(n+2)}}$'.
When we explicitly evaluate the P\"oppe products in the monomials
shown on the right in this expression for $\tp\pp_n$, we 
generate linear combinations of compositions in $\Cb(n+2)$.
Just like in the proof of Theorem~\ref{thm:mainresult}, 
imagine constructing a table whose rows are parametrised
by the set of compositions $\Cb(n+2)$, ordered according to descent order,
and whose columns are parametrised by the set of compositions $\Cb^\ast(n+2)$,
also ordered according to descent order.
Our goal is to prove there exists a unique set of real coefficients
$\{\tp c_w\colon w\in\Cb^\ast(n+2)\backslash(n+2)\}$ such that in $\Rb\la\Cb\ra_\ast$, we have 
\begin{equation*}
\sum_{k=1}^{\frac12(n+3)}\sum_{a_1\cdots a_k\in\Cb^\ast(n+2,k)}
\tp c_{a_1\cdots a_k}\cdot\wdh{\und{a}}_1\ast\cdots\ast\wdh{\und{a}}_k=1(n+1)+(n+1)1+2\cdot(1n1).
\end{equation*}
We nominate the final column in the table to represent the right-hand side in
the equation above which is $1(n+1)+(n+1)1+2\cdot(1n1)\in\Rb\la\Cb\ra_\ast$.
Just as in the proof of Theorem~\ref{thm:mainresult}, the table represents
a linear equation of the form
\begin{equation*}
\tp A\,\tp\cb=\tp\bb,
\end{equation*}
where $\tp\cb$ is the vector of unknown coefficients $\{\tp c_w\colon w\in\Cb^\ast(n+2)\backslash(n+2)\}$,
in descent order so $\cb=(\tp c_{n1},\tp c_{(n-1)2},\ldots,\tp c_{1n},c_{(n-2)11},\ldots)^{\mathrm{T}}$.
The vector $\tp\bb$ contains all zeros apart from the rows corresponding to
the compositions $(n+1)1$, $1(n+1)$ and $1n1$, where it contains the
respective real values $1$, $1$ and $2$.
As in the proof of Theorem~\ref{thm:mainresult} we decompose the rows of $\tp A$ into
blocks of compositions in $\Cb(n+2,k')$ for $k'=2,\ldots,n$, and we can
decompose columns of $\tp A$ into blocks of compositions in
$\Cb^\ast(n-k+3,k)$ for $k=2,\ldots,\frac12(n+3)$.
In terms of blocks, the matrix $\tp A$ is lower block triangular.
We proceed block by block. 
To begin, we focus on the $2$-part composition sub-block of $\tp A$ parametrised
by the columns $\Cb^\ast(n+1,2)$ and by the rows $\Cb(n+2,2)$, both in descent order.
The entries in $\tp A$ are equivalent to those shown in Table~\ref{table:secondblock},
once we neglect the first column and replace $n$ by $n+2$ therein.
This block has size $(n+1)\times n$. The first equation in this block is $\tp c_{n1}=1$.
However, we temporarily ignore this equation and the information it carries.
The remaining $n\times n$ block whose rows correspond to the compositions
$n2$, $(n-1)3$, \ldots, $1(n+1)$, contains $(n-1)$ homogeneous linear equations
for $\ell=2,\ldots,n$, given by
\begin{equation*}
\tp c_{(n-\ell+2)(\ell-1)}+\tp c_{(n-\ell+1)\ell}=0,
\end{equation*}
and a final non-homogenous linear equation $\tp c_{1n}=1$.
Note, all the signature coefficients are equal to $1$.
Due to the inherent symmetry in the relation between the terms given by
$1(n+1)+(n+1)1+2\cdot(1n1)$
and $\tp\pp_n$ proposed above, we \emph{replace} this final equation by the
homogensous linear equation $-\tp c_{n1}+\tp c_{1n}=0$. 
The homogeneous linear system of equations represent by those shown above  
for $\ell=2,\ldots,n$ can by straightforwardly solved, systematically working through
the equations in order, giving $\tp c_{(n-\ell+1)\ell}=(-1)^{(\ell-1)}\tp c_{n1}$.
Note when $\ell=n$ we have $\tp c_{1n}=(-1)^{n-1}\tp c_{n1}=(-1)^{2m}\tp c_{n1}=\tp c_{n1}$
when $n=2m+1$ is odd. This is thus consistent with the replacement we introduced above.
We henceforth retain the homogeneous linear system of equations for $\ell=2,\ldots,n$ shown above,
and the final replacement homogeneous linear equation.
We nominate the ``free'' variable to be $\tp c_{n1}$. 
This resulting $n\times n$ system of homogeneous linear equations is upper triangular
with all the leading diagonal entries non-zero. 
Thus all these leading diagonal entries can be used as pivots
in a Gaussian elimination process to render all the entries in the corresponding
columns in the rows below this block to be zero.

We now adopt a strategy for solution for the linear system $\tp A\,\tp\cb=\tp\bb$,
analogous to that in the proof of Theorem~\ref{thm:mainresult}.
However, before proceeding there is one minor \emph{snag} in the procedure we must deal
with first. The snag is that all the remaining linear equations are homogeneous apart from
one, namely that for the row corresponding to the composition $1n1$. This non-homogeneous
linear equation has the form 
\begin{equation*}
\chi_{1(n-1)\smb1}\,\tp c_{n1}+\chi_{1\smb(n-1)1}\,\tp c_{1n}
+\chi_{1\smb(n-2)\smb1}\,\tp c_{1(n-2)1}=2.
\end{equation*}
Note hereafter in this proof we denote the arguments of the signature character maps $\chi$
as sub-indicies, as indicated in the equation above, for brevity.
Since we know that $\chi_{1(n-1)\smb1}=\chi_{1\smb(n-1)1}=\chi_{1(n-1)}\,\chi_1=n$
and $\chi_{1\smb(n-2)\smb1}=1$, the equation above is equivalent to
$n\,\tp c_{n1}+n\,\tp c_{1n}+\tp c_{1(n-2)1}=2$.
This equation is in turn equivalent to the relation $\tp c_{1(n-2)1}=2-2n\,\tp c_{1n}$.
We now demonstrate that this equation is consistent with the preceding pivot equations,
in terms of descent order. Recall our dicussion involving the level~$2$, $3$-part composition sub-blocks
of $A$ in the proof of Theorem~\ref{thm:mainresult}. The pivot equation in the row
preceding the row $1n1$ for the equation shown above occurs in the row corresponding to the
composition $2(n-1)1$. The pivot equation in that row, i.e.\/ in the row corresponding to $2(n-1)1$,
involves the unknown $\tp c_{2(2-3)1}$, and so we must invoke the pivot equation corresponding
to that variable that lies in a preceding row. Iteratively continuing this procedure
we arrive at the following system of homogeneous linear equations respectively given by the
rows corresponding to the compositions $2(n-1)1$, $3(n-2)1$, $4(n-3)1$, \ldots, $(n-2)11$:
\begin{align*}
\tp c_{1(n-2)1}+\tp c_{2(n-3)1}+\chi_{(n-2)\smb1}\,\tp c_{n1}
+\chi_{2\smb(n-2)1}\,\tp c_{2(n-1)}+\chi_{1\smb(n-1)1}\,\tp c_{1n}&=0,\\
\tp c_{2(n-3)1}+\tp c_{3(n-4)1}+\chi_{(n-3)\smb1}\,\tp c_{n1}
+\chi_{3\smb(n-3)1}\,\tp c_{3(n-2)}+\chi_{2\smb(n-2)1}\,\tp c_{2(n-1)}&=0,\\
\tp c_{3(n-4)1}+\tp c_{4(n-5)1}+\chi_{(n-4)\smb1}\,\tp c_{n1}
+\chi_{4\smb(n-4)1}\,\tp c_{4(n-3)}+\chi_{3\smb(n-3)1}\,\tp c_{3(n-2)}&=0,\\
\vdots &  \\
\tp c_{(n-3)21}+\tp c_{(n-2)11}+\chi_{(n-2)2\smb1}\,\tp c_{n1}
+\chi_{(n-2)\smb21}\,\tp c_{(n-2)3}+\chi_{(n-3)\smb31}\,\tp c_{(n-3)4}&=0,\\
\tp c_{(n-2)11}+\chi_{(n-1)1\smb1}\,\tp c_{n1}
+\chi_{(n-1)\smb11}\,\tp c_{(n-1)2}+\chi_{(n-2)\smb21}\,\tp c_{(n-2)3}&=0.
\end{align*}
We now utilise that for $\ell=2,\ldots,n$ we know $\tp c_{(n-\ell+1)\ell}=(-1)^{(\ell-1)}\tp c_{n1}$.
Solving the homogeneous linear system of equations above in reverse we observe: 
\begin{align*}
\tp c_{(n-2)11}&=\bigl(\chi_{(n-1)1}-1\bigr)\,\tp c_{n1},\\
\tp c_{(n-3)21}&=\bigl(\chi_{(n-1)1}-\chi_{(n-2)2}+2\bigr)\,\tp c_{n1},\\
\vdots &  \\
\tp c_{1(n-2)1}&=\bigl(\chi_{(n-1)1}
+\cdots+(-1)^{n-1}\chi_{2(n-2)}+(-1)^{n-2}(n-2)\bigr)\,\tp c_{n1}.
\end{align*}
We observe that the coefficient of $\tp c_{n1}$ on the right for odd $n$ is given by,
\begin{align*}
\chi_{(n-1)1}-&\chi_{(n-2)2}+\cdots+(-1)^{n-1}\chi_{2(n-2)}+(-1)^{n-2}(n-2)\\
&=\begin{pmatrix}n\\(n-1)\end{pmatrix}-\begin{pmatrix}n\\(n-2)\end{pmatrix}
+\cdots+(-1)^{n-1}\begin{pmatrix}n\\2\end{pmatrix}+(-1)^{n-2}(n-2)\\
&=-\Bigl((1-x)^n-1-n(-x)^{n-1}-(-x)^n\Bigr)\Big|_{x=1}-(n-2)\\
&=2-2n.
\end{align*}
Hence we observe that we have the two expressions resulting from the row corresponding
to the composition $1n1$, namely $\tp c_{1(n-2)1}=2-2n\,\tp c_{1n}$ and $\tp c_{1(n-2)1}=(2-2n)\,\tp c_{1n}$.
Temporarily holding these two equations back, we observe that we can proceed to determine all
of the coefficients $\tp c_w$ for $w\in\Cb^\ast(n+1)\backslash(n+2)$ in terms of $\tp c_{n1}$
by precisely following the procedure outlined in the proof of Theorem~\ref{thm:mainresult}.
Finally we recover the first equation, i.e.\/ that $\tp c_{n1}=1$, which determines
all the coefficients $\tp c_w$ uniquely, and shows the two equations we temporarily held
back are consistent. \qed
\end{proof}
\begin{remark}[The Lax hierarchy is unique]
The hierarchy given in Theorem~\ref{thm:mainresult} at each odd order $n$
has a field given by a unique polynomial
$\pi_n=\pi_n\bigl(\wdh{\und{1}},\wdh{\und{2}},\ldots,\wdh{\und{(n-2)}},\wdh{\und{n}}\bigr)$
in $\Rb\la\Cb\ra_\ast$. For the Lax hierarchy presented in
Theorem~\ref{thm:Laxhierarchy} the vector field has the polynomial form
$\wdh{\und{n}}+\hat\pp_n\bigl(\wdh{\und{1}},\wdh{\und{2}},\ldots,\wdh{\und{(n-2)}}\bigr)$
in $\Rb\la\Cb\ra_\ast$. The uniqueness of the former polynomial
guarantees the two hierarchies are uniquely one and the same.  
\end{remark}

\section{Conclusion}\label{sec:conclusion}
We have shown that integrability in the sense of linearisation for the
whole Korteweg--de Vries hierarchy is equivalent to the existence of
polynomial expansions for basic compositions in the real algebra of
compositions equipped with the P\"oppe product, $\Rb\la\Cb\ra_\ast$.
This procedure opens many new research directions we intend to pursue next:
\begin{enumerate}
\item[(i)] Do all compositions and/or linear combinations of compositions
in $\Rb\la\Cb\ra_\ast$ have P\"oppe polynomial expansions?
\item[(ii)] There are natural connections to shuffle algebras, Rota--Baxter operators,
symmetric functions and so forth, do these provide any further insight into integrability
in this context?
\item[(iii)] The P\"oppe product can be considered as a quasi-Leibniz product with the
third term in Definition~\ref{def:Poppeproductcompositions}
considered to be the modification from the Leibniz product.
This term is analogous to the the modification distinguishing the quasi-shuffle product
from the shuffle product, see for example Curry \textit{et al.\/} \cite{CE-FMW}. 
A natural question is, how robust are the results established
in Sections~\ref{sec:existenceanduniqueness} and \ref{sec:Laxhierarchy}
to perturbation of the coefficient `$2$' in front of the third modification term
in the P\"oppe product?
\item[(iv)] The shuffle product and de-shuffle co-product are adjoint operators, as are
the concatenation product and deconcatenation, see Reutenauer~\cite[p.~27]{Reu}.
Can we prove that the P\"oppe product and de-P\"oppe co-product are adjoint operations?
\item[(v)] Following on from (iv), and as hinted at in Section~\ref{sec:signaturecoalgebra},
a natural progression is to introduce the following tensor product of algebras: 
$\Sb\otimes\Rb\la\Cb\ra_\ast$.
Note we have $\Sb\cong\Rb\la\Cb^\otimes\ra$.
The idea is to pull back linear combinations and/or polynomials 
in $\Rb\la\Cb\ra_\ast\cong\chi(\Sb)\otimes\Rb\la\Cb\ra_\ast$ to
$\Sb\otimes\Rb\la\Cb\ra_\ast$ and then perform computations in this
latter algebra. Then further natural questions arise concerning whether
it is possible to endow $\Sb$ with a compatible product and $\Rb\la\Cb^\otimes\ra$
a compatible co-product to raise their status to bi-algebras and perhaps
even Hopf algebras?
\item[(vi)] The generation of the Korteweg--de Vries hierarchy via the Lax iteration
shown in Theorem~\ref{thm:Laxhierarchy} involves evaluating
a second order derivation $\md^2$ together with the other terms shown.
We know from Theorem~\ref{thm:Laxhierarchy} that overall, the terms
generated via the Lax iteration can be expressed in the form of a
polynomial in the signature expansions. Further for example, the
nonlinear terms can be expressed in terms of derivations as follows.
For the quintic non-commutative potential Korteweg--de Vries
equation, the quadratic nonlinear terms can be encoded as
$5\cdot\bigl(\md^2(\und{1}\ast\und{1})-\und{2}\ast\und{2}\bigr)$.
For the septic version in the hierarchy the quadratic nonlinear terms
can be encoded as
$7\cdot\bigl(\md^4(\und{1}\ast\und{1})-2\cdot\md^2(\und{2}\ast\und{2})+\und{3}\ast\und{3}\bigr)$.
The question is, can the cubic and higher degree terms be similarly encoded
and this formulation of the nonlinear terms be propagated to all orders?
\item[(vii)] Finally, crucially, Doikou \text{et al.\/} \cite{DMS20}
prove that the non-commutative nonlinear Schr\"o-dinger and modified Korteweg--de Vries
flows are Grassmannian flows and integrable in this sense.
Indeed we show that nonlocal reverse time and reverse space-time versions
of these equations are also Grassmannian flows.
See Ablowitz and Musslimani~\cite{AM},
G\"urses and Pekcan~\cite{GP2019a,GP2019b,GP2018,GP2020} and Fokas~\ref{Fokas2016}
for more details on such systems.
Further, Malham~\cite{M-quinticNLS} demonstrated that
the local and nonlocal non-commutative fourth order quintic
nonlinear Schr\"odinger flow is also a Grassmannian flow. A natural
extension is to determine if the combinatorial algebraic procedure we develop
herein can be extended, in modified form, to all orders in the classical hierarchies for such
non-commutative integrable systems? A modification is necessary as 
the linear system of equations required to prescribe the Grassmannian flow
for these systems has the form: $\mathcal P=\mathcal G(\id+\mathcal Q)$,
where
\begin{equation*}
\mathcal P=\begin{pmatrix}P\\\tilde P\end{pmatrix},\qquad
\mathcal Q=\begin{pmatrix} Q & O \\ O & \tilde Q\end{pmatrix}
\qquad\text{and}\qquad
\mathcal G=\begin{pmatrix}G\\\tilde G\end{pmatrix}.
\end{equation*}
Here the block operators shown satisfy the following linear system of equations:
$\pa_tP=d(\pa)P$, $\pa_t\tilde P=\tilde d(\pa)\tilde P$, 
$Q\coloneqq\tilde PP$ and $\tilde Q\coloneqq P\tilde P$, 
and $d$ as well as $\tilde d$ are suitable constant coefficient polynomials in $\pa$.
Note once the operators $P$ and $\tilde P$ are determined as solutions to the
linear partial differential systems shown, the operators $Q$ and $\tilde Q$ are
prescibed by the quadratic forms shown. It is these definitions for $Q$ and $\tilde Q$
that necessitate, for the fully general case, a non-trivial extension of the
combinatorial algebra introduced herein.
\end{enumerate}

%
%

\begin{acknowledgement}
SJAM thanks Ioannis Stylianidis for some very helpful discussions.
\end{acknowledgement}

\appendix

\section{Signature co-algebra proof}\label{sec:co-algebraproof}
We prove Theorem~\ref{thm:sigcoalg}.
First, let us prove axiom (i). Since the co-unit given in Definition~\ref{def:signaturecoalgebra}
sends all non-empty words to $0\in\Rb$ and the empty word $\nu$ to $1\in\Rb$, if we apply  
$\id\otimes\varepsilon$ to the form for $\Delta(a_1a_2\cdots a_n)$ given in Lemma~\ref{lemma:co-prod},
we observe the only non-zero image would be from the final term for which
$(\id\otimes\varepsilon)\circ(a_1\cdots a_n\otimes\nu)=a_1\cdots a_n=\id\circ(a_1\cdots a_n)$.
Likewise the only non-zero image of $\varepsilon\otimes\id$ would be from the first
term for which
$(\varepsilon\otimes\id)\circ(\nu\otimes a_1\cdots a_n)=a_1\cdots a_n=\id\circ(a_1\cdots a_n)$.
Hence axiom (i) is satisfied.

Second, we focus on proving axiom (ii). It helps to consider 
the co-product $\delta\colon\Cb\to\Cb\otimes\Cb$ which denotes the standard
deconcatenation co-product, i.e.\/ for any composition $a_1a_2\cdots a_n\in\Cb$ we have  
\begin{equation*}
\delta(a_1a_2\cdots a_n)=\sum_{k=0}^n a_1\cdots a_k\otimes a_{k+1}\cdots a_n,
\end{equation*}
where we use the convention that when $k=0$ and $k=n$, the corresponding
terms are respectively $\nu\otimes a_1a_2\cdots a_n$ and $a_1\cdots a_n\otimes\nu$.
Note we have $\delta(\nu)=\nu\otimes\nu$.
We observe that if we apply $(\id\otimes\delta)$ to the last expression we get
\begin{align*}
(\id\otimes\delta)\circ\delta(a_1a_2\cdots a_n)
&=\sum_{k=0}^n a_1\cdots a_k\otimes \delta(a_{k+1}\cdots a_n)\\
&=\sum_{k=0}^n a_1\cdots a_k\otimes
\Biggl(\sum_{\ell=0}^{n-k}a_{k+1}\cdots a_{k+\ell}\otimes a_{k+\ell+1}\cdots a_n)\Biggr)\\
&=\sum_{k=0}^n\sum_{\ell=0}^{n-k}
a_1\cdots a_k\otimes a_{k+1}\cdots a_{k+\ell}\otimes a_{k+\ell+1}\cdots a_n.
\end{align*}
We observe the last line represents the sum over all possible ways to deconcatenate
$a_1\cdots a_n$ into three parts, by performing a ``left-to-right'' secondary deconcatenation
after the first. Hence we expect the action of $(\delta\otimes\id)$ on $\delta(a_1a_2\cdots a_n)$
to correspond to performing ``right-to-left'' secondary deconcatenation after the first,
producing the same triple deconcatenation result. Indeed we observe,
\begin{align*}
(\delta\otimes\id)\circ\delta(a_1a_2\cdots a_n)
&=\sum_{k=0}^n \delta(a_1\cdots a_k)\otimes a_{k+1}\cdots a_n\\
&=\sum_{k=0}^n\Biggl(\sum_{\ell=0}^ka_1\cdots a_{\ell}\otimes a_{\ell+1}\cdots a_k\Biggr)\otimes a_{k+1}\cdots a_n\\
&=\sum_{k=0}^n\sum_{\ell=0}^k
a_1\cdots a_{\ell}\otimes a_{\ell+1}\cdots a_k\otimes a_{k+1}\cdots a_n\\
&=\sum_{\ell=0}^n\sum_{k=\ell}^n
a_1\cdots a_{\ell}\otimes a_{\ell+1}\cdots a_k\otimes a_{k+1}\cdots a_n\\
&=\sum_{\ell=0}^n\sum_{k=0}^{n-\ell}
a_1\cdots a_{\ell}\otimes a_{\ell+1}\cdots a_{k+\ell}\otimes a_{k+\ell+1}\cdots a_n,
\end{align*}
which is the same as the ``left-to-right'' sum once we swap the identities of
the dummy summation variables $k$ and $\ell$. The identification of $(\id\otimes\Delta)\circ\Delta$
and $(\Delta\otimes\id)\circ\Delta$ follows in exactly the same way.
For the $\Delta(a_1\cdots a_n)$ co-product expansion,
except for the $\nu\otimes a_1\cdots a_n$ and $a_1\cdots a_n\otimes\nu$ terms,
we replace the terms $a_1\cdots a_k\otimes a_{k+1}\cdots a_n$ in $\delta(a_1\cdots a_n)$
for $k=1,\ldots,n-1$ by the pair of terms 
\begin{equation*}
a_1\cdots a_{k-1}(a_k-1)\otimes a_{k+1}\cdots a_n+a_1\cdots a_{k}\otimes(a_{k+1}-1)a_{k+2}\cdots a_n.
\end{equation*}
The co-product $\Delta(a_1\cdots a_n)$ tensorally decomposes
$a_1\cdots a_n$ into the sum of all possible composition pairs that
produce $a_1\cdots a_n$ via the P\"oppe product $\ast$, including the 
empty composition $\nu$. Hence $(\id\otimes\Delta)\circ\Delta$ will tensorally
decompose $a_1\cdots a_n$ into the sum of all possible composition triples $u\otimes v\otimes w$
that produce $a_1\cdots a_n$ via the P\"oppe product $u\ast v\ast w$.
The action of $(\id\otimes\Delta)\circ\Delta$ is to perform a ``left-to-right'' 
secondary de-P\"oppe co-product after a first application of the de-P\"oppe co-product.
The action of $(\Delta\otimes\id)\circ\Delta$ is to perform a ``right-to-left'' 
secondary de-P\"oppe co-product after a first application of the de-P\"oppe co-product.
Both generate the same sum over all triples $u\otimes v\otimes w$ that produce $a_1\cdots a_n$.
As in the proof of Lemma~\ref{lemma:co-prod}, care must be taken to account for
instances when $a_k=1$. These can be taken care of by a post-application of
$\theta\otimes\theta\otimes\theta$ and extending the definition of $\theta$ to include
the circumstance $\theta\colon -\boldsymbol 1\mapsto0\,\nu$ where in the image
$0\in\Rb$ so the corresponding term is eliminated. The expression $-\boldsymbol 1$ could
arise through a double application of $\md^{-1}$ to a letter $a=1$.
We have thus established that axiom (ii) is satisfied.
The proof is complete.

\section{Tables}\label{sec:tables}

\textheight=23cm

\begin{landscape}
\begin{table}
\caption{Non-zero signature coefficients in $\pi_n$ appearing in the $3$-part composition blocks of $A$.
The coefficients are the $\chi$-images of the signature entries shown.}
\label{table:KdVn-3parts}
\begin{center}
\hspace{-1cm}
\begin{tabular}{||l||c|cc|ccc|cccc|cc|}
\hline\hline
$\!\phantom{\biggl|}\Cb\!$
&$\!\!\und{(n\!-\!4)}\scriptsize{\ast}\und{1}\scriptsize{\ast}\und{1}\!\!$
&$\!\!\und{(n\!-\!5)}\scriptsize{\ast}\und{2}\scriptsize{\ast}\und{1}\!\!$
&$\!\!\und{(n\!-\!5)}\scriptsize{\ast}\und{1}\scriptsize{\ast}\und{2}\!\!$
&$\!\!\und{(n\!-\!6)}\scriptsize{\ast}\und{3}\scriptsize{\ast}\und{1}\!\!$
&$\!\!\und{(n\!-\!6)}\scriptsize{\ast}\und{2}\scriptsize{\ast}\und{2}\!\!$
&$\!\!\und{(n\!-\!6)}\scriptsize{\ast}\und{1}\scriptsize{\ast}\und{3}\!\!$
&$\!\!\und{(n\!-\!7)}\scriptsize{\ast}\und{4}\scriptsize{\ast}\und{1}\!\!$
&$\!\!\und{(n\!-\!7)}\scriptsize{\ast}\und{3}\scriptsize{\ast}\und{2}\!\!$
&$\!\!\und{(n\!-\!7)}\scriptsize{\ast}\und{2}\scriptsize{\ast}\und{3}\!\!$
&$\!\!\und{(n\!-\!7)}\scriptsize{\ast}\und{1}\scriptsize{\ast}\und{4}\!\!$
&$\!\!\und{(n\!-\!8)}\scriptsize{\ast}\und{5}\scriptsize{\ast}\und{1}\!\!$
&$\cdots$\\
\hline
\hline
$\!\!(n-2)11\!\!$ &&&&&&&&&&&&\\
\hline
$\!\!(n-3)21\!\!$ & $\!\!(n-4)\smb1\smb1\!\!$ & &&&&&&&&&&\\
$\!\!(n-3)12\!\!$ & $\!\!(n-4)\smb1\smb1\!\!$ & &&&&&&&&&&\\
\hline
$\!\!(n-4)31\!\!$ & $\!\!(n-4)\smb1\smb1\!\!$ & $\!\!(n-5)\smb2\smb1\!\!$ &&&&&&&&&&\\
$\!\!(n-4)22\!\!$ & $\!\!(n-4)\smb1\smb1\!\!$ & $\!\!(n-5)\smb2\smb1\!\!$ & $\!\!(n-5)\smb1\smb2\!\!$ &&&&&&&&&\\
$\!\!(n-4)13\!\!$ &                 &                 & $\!\!(n-5)\smb1\smb2\!\!$ &&&&&&&&&\\
\hline
$\!\!(n-5)41\!\!$ && $\!\!(n-5)\smb2\smb1\!\!$ &                 & $\!\!(n-6)\smb3\smb1\!\!$ &&&&&&&&\\
$\!\!(n-5)32\!\!$ && $\!\!(n-5)\smb2\smb1\!\!$ & $\!\!(n-5)\smb1\smb2\!\!$ & $\!\!(n-6)\smb3\smb1\!\!$ & $\!\!(n-6)\smb2\smb2\!\!$ &&&&&&&\\
$\!\!(n-5)23\!\!$ &&& $\!\!(n-5)\smb1\smb2\!\!$ && $\!\!(n-6)\smb2\smb2\!\!$ & $\!\!(n-6)\smb1\smb3\!\!$ &&&&&&\\
$\!\!(n-5)14\!\!$ &&&         &                 && $\!\!(n-6)\smb1\smb3\!\!$ &&&&&&\\
\hline
$\!\!(n-6)51\!\!$ &&&& $\!\!(n-6)\smb3\smb1\!\!$ &&& $\!\!(n-7)\smb4\smb1\!\!$ &&&&&\\
$\!\!(n-6)42\!\!$ &&&& $\!\!(n-6)\smb3\smb1\!\!$ & $\!\!(n-6)\smb2\smb2\!\!$ && $\!\!(n-7)\smb4\smb1\!\!$ & $\!\!(n-7)\smb3\smb2\!\!$ &&&&\\
$\!\!(n-6)33\!\!$ &&&&& $\!\!(n-6)\smb2\smb2\!\!$ & $\!\!(n-6)\smb1\smb3\!\!$ && $\!\!(n-7)\smb3\smb2\!\!$ & $\!\!(n-7)\smb2\smb3\!\!$ &&&\\
$\!\!(n-6)24\!\!$ &&&&&& $\!\!(n-6)\smb1\smb3\!\!$ &&& $\!\!(n-7)\smb2\smb3\!\!$ & $\!\!(n-7)\smb1\smb4\!\!$ &&\\
$\!\!(n-6)15\!\!$ &&&&&&&&&& $\!\!(n-7)\smb1\smb4\!\!$ &&\\
\hline
$\!\!(n-7)61\!\!$ &&&&&&& $\!\!(n-7)\smb4\smb1\!\!$ &&&& $\!\!(n-8)\smb5\smb1\!\!$ &\\
$\!\!(n-7)52\!\!$ &&&&&&& $\!\!(n-7)\smb4\smb1\!\!$ & $\!\!(n-7)\smb3\smb2\!\!$ &&& $\!\!(n-8)\smb5\smb1\!\!$ &\\
$\!\!(n-7)43\!\!$ &&&&&&&& $\!\!(n-7)\smb3\smb2\!\!$ & $\!\!(n-7)\smb2\smb3\!\!$ &&&\\
$\vdots$ &&&&&&&&&&&&\\
\hline\hline
\end{tabular}
\end{center}
\end{table}
\end{landscape}

\textheight=25cm

\begin{landscape}
\begin{table}
\caption{Non-zero signature coefficients in $\pi_n$ appearing in the $4$-part composition blocks of $A$.
The coefficients are the $\chi$-images of the signature entries shown.}
\label{table:KdVn-4parts}
\begin{center}
\hspace{-1cm}
\begin{tabular}{||l||c|ccc|cccccc|c|}
\hline\hline
$\!\phantom{\biggl|}\Cb\!$
&$\!\!\und{(n\!-\!6)}\scriptsize{\ast}\und{1}\scriptsize{\ast}\und{1}\scriptsize{\ast}\und{1}\!\!$
&$\!\!\und{(n\!-\!7)}\scriptsize{\ast}\und{2}\scriptsize{\ast}\und{1}\scriptsize{\ast}\und{1}\!\!$
&$\!\!\und{(n\!-\!7)}\scriptsize{\ast}\und{1}\scriptsize{\ast}\und{2}\scriptsize{\ast}\und{1}\!\!$
&$\!\!\und{(n\!-\!7)}\scriptsize{\ast}\und{1}\scriptsize{\ast}\und{1}\scriptsize{\ast}\und{2}\!\!$
&$\!\!\und{(n\!-\!8)}\scriptsize{\ast}\und{3}\scriptsize{\ast}\und{1}\scriptsize{\ast}\und{1}\!\!$
&$\!\!\und{(n\!-\!8)}\scriptsize{\ast}\und{2}\scriptsize{\ast}\und{2}\scriptsize{\ast}\und{1}\!\!$
&$\!\!\und{(n\!-\!8)}\scriptsize{\ast}\und{2}\scriptsize{\ast}\und{1}\scriptsize{\ast}\und{2}\!\!$
&$\!\!\und{(n\!-\!8)}\scriptsize{\ast}\und{1}\scriptsize{\ast}\und{3}\scriptsize{\ast}\und{1}\!\!$
&$\!\!\und{(n\!-\!8)}\scriptsize{\ast}\und{1}\scriptsize{\ast}\und{2}\scriptsize{\ast}\und{2}\!\!$
&$\!\!\und{(n\!-\!8)}\scriptsize{\ast}\und{1}\scriptsize{\ast}\und{1}\scriptsize{\ast}\und{3}\!\!$
&$\cdots$\\
\hline
\hline
$\!\!(n-3)111\!\!$ &&&&&&&&&&&\\
\hline
$\!\!(n-4)211\!\!$ &&&&&&&&&&&\\
$\!\!(n-4)121\!\!$ &&&&&&&&&&&\\
$\!\!(n-4)112\!\!$ &&&&&&&&&&&\\
\hline
$\!\!(n-5)311\!\!$ &&&&&&&&&&&\\
$\!\!(n-5)221\!\!$ & $\!\!(n-6)\smb1\smb1\smb1\!\!$ &&&&&&&&&&\\
$\!\!(n-5)212\!\!$ & $\!\!(n-6)\smb1\smb1\smb1\!\!$ &&&&&&&&&&\\
$\!\!(n-5)131\!\!$ & $\!\!(n-6)\smb1\smb1\smb1\!\!$ &&&&&&&&&&\\
$\!\!(n-5)122\!\!$ & $\!\!(n-6)\smb1\smb1\smb1\!\!$ &&&&&&&&&&\\
$\!\!(n-5)113\!\!$ &&&&&&&&&&&\\
\hline
$\!\!(n-6)411\!\!$ &&&&&&&&&&&\\
$\!\!(n-6)321\!\!$ & $\!\!(n-6)\smb1\smb1\smb1\!\!$ & $\!\!(n-7)\smb2\smb1\smb1\!\!$ &&&&&&&&&\\
$\!\!(n-6)312\!\!$ & $\!\!(n-6)\smb1\smb1\smb1\!\!$ & $\!\!(n-7)\smb2\smb1\smb1\!\!$ &&&&&&&&&\\
$\!\!(n-6)231\!\!$ & $\!\!(n-6)\smb1\smb1\smb1\!\!$ & $\!\!(n-7)\smb2\smb1\smb1\!\!$ & $\!\!(n-7)\smb1\smb2\smb1\!\!$ &&&&&&&&\\
$\!\!(n-6)222\!\!$ & $\!\!(n-6)\smb1\smb1\smb1\!\!$ & $\!\!(n-7)\smb2\smb1\smb1\!\!$ & $\!\!(n-7)\smb1\smb2\smb1\!\!$ & $\!\!(n-7)\smb1\smb1\smb2\!\!$ &&&&&&&\\
$\!\!(n-6)213\!\!$ &&&& $\!\!(n-7)\smb1\smb1\smb2\!\!$ &&&&&&& \\
$\!\!(n-6)141\!\!$ &&& $\!\!(n-7)\smb1\smb2\smb1\!\!$ &&&&&&&&\\
$\!\!(n-6)132\!\!$ &&& $\!\!(n-7)\smb1\smb2\smb1\!\!$ & $\!\!(n-7)\smb1\smb1\smb2\!\!$ &&&&&&& \\
$\!\!(n-6)123\!\!$ &&&& $\!\!(n-7)\smb1\smb1\smb2\!\!$ &&&&&&& \\
$\!\!(n-6)114\!\!$ &&&&&&&&&&&\\
\hline
$\!\!(n-7)511\!\!$ &&&&&&&&&&&\\
$\!\!(n-7)421\!\!$ && $\!\!(n-7)\smb2\smb1\smb1\!\!$ &&& $\!\!(n-8)\smb3\smb1\smb1\!\!$ &&&&&&\\
$\!\!(n-7)412\!\!$ && $\!\!(n-7)\smb2\smb1\smb1\!\!$ &&& $\!\!(n-8)\smb3\smb1\smb1\!\!$ &&&&&&\\
$\!\!(n-7)331\!\!$ && $\!\!(n-7)\smb2\smb1\smb1\!\!$ & $\!\!(n-7)\smb1\smb2\smb1\!\!$ && $\!\!(n-8)\smb3\smb1\smb1\!\!$ & $\!\!(n-8)\smb2\smb2\smb1\!\!$ &&&&&\\
$\!\!(n-7)322\!\!$ && $\!\!(n-7)\smb2\smb1\smb1\!\!$ & $\!\!(n-7)\smb1\smb2\smb1\!\!$ & $\!\!(n-7)\smb1\smb1\smb2\!\!$ & $\!\!(n-8)\smb3\smb1\smb1\!\!$ & $\!\!(n-8)\smb2\smb2\smb1\!\!$ & $\!\!(n-8)\smb2\smb1\smb2\!\!$ &&&&\\
$\!\!(n-7)313\!\!$ &&&& $\!\!(n-7)\smb1\smb1\smb2\!\!$ &&& $\!\!(n-8)\smb2\smb1\smb2\!\!$ &&&&\\
$\!\!(n-7)241\!\!$ &&& $\!\!(n-7)\smb1\smb2\smb1\!\!$ &&& $\!\!(n-8)\smb2\smb2\smb1\!\!$ && $\!\!(n-8)\smb1\smb3\smb1\!\!$ &&& \\
$\!\!(n-7)232\!\!$ &&& $\!\!(n-7)\smb1\smb2\smb1\!\!$ & $\!\!(n-7)\smb1\smb1\smb2\!\!$ && $\!\!(n-8)\smb2\smb2\smb1\!\!$ & $\!\!(n-8)\smb2\smb1\smb2\!\!$ & $\!\!(n-8)\smb1\smb3\smb1\!\!$ & $\!\!(n-8)\smb1\smb2\smb2\!\!$ && \\
$\!\!(n-7)223\!\!$ &&&& $\!\!(n-7)\smb1\smb1\smb2\!\!$ &&& $\!\!(n-8)\smb2\smb1\smb2\!\!$ && $\!\!(n-8)\smb1\smb2\smb2\!\!$ & $\!\!(n-8)\smb1\smb1\smb3\!\!$ & \\
$\!\!(n-7)214\!\!$ &&&&&&&&&& $\!\!(n-8)\smb1\smb1\smb3\!\!$ &\\
$\vdots$ &&&&&&&&&&&\\
\hline\hline
\end{tabular}
\end{center}
\end{table}
\end{landscape}

\end{document}